\documentclass[11pt]{article}

\usepackage{fullpage}

\usepackage{dsfont}
\usepackage{latexsym}
\usepackage{natbib}
\usepackage{amsmath}
\usepackage{amsthm}
\usepackage{amssymb}
\usepackage{amsfonts}
\usepackage{mathrsfs}
\usepackage{aliascnt}
\usepackage{pstricks}
\usepackage{pst-all}
\usepackage{pstricks-add}
\usepackage{pst-plot}
\usepackage{subfig}
\usepackage{graphicx}
\usepackage[pdfpagelabels,pdfpagemode=None]{hyperref}
\usepackage{algorithmic}
\usepackage{algorithm}

\usepackage{enumitem}
\setlist[enumerate]{label=(\alph*), ref=(\alph*)}

\newcommand{\argmax}{\operatorname{arg\,max}}

\newtheorem{theorem}{Theorem}[section]%

\newaliascnt{lemma}{theorem}
\newtheorem{lemma}[lemma]{Lemma}%
\aliascntresetthe{lemma}

\newaliascnt{claim}{theorem}
\newtheorem{claim}[claim]{Claim}%
\aliascntresetthe{claim}

\newaliascnt{corollary}{theorem}
\newtheorem{corollary}[corollary]{Corollary}%
\aliascntresetthe{corollary}

\newaliascnt{proposition}{theorem}
\newtheorem{proposition}[proposition]{Proposition}%
\aliascntresetthe{proposition}

\newaliascnt{remark}{theorem}
\theoremstyle{definition}

\aliascntresetthe{remark}

\newaliascnt{algo}{procedure}
\aliascntresetthe{algo}

\theoremstyle{definition}
\newtheorem{definition}{Definition}[section]

\newtheorem{example}{Example}[section]

\newcommand{\AutoAdjust}[3]{\mathchoice{ \left #1 #2  \right #3}{#1 #2 #3}{#1 #2 #3}{#1 #2 #3} }
\newcommand{\Xcomment}[1]{{}}

\newcommand{\InBrackets}[1]{\AutoAdjust{[}{#1}{]}}
\newcommand{\Ex}[2][]{\operatorname{\mathbf E}_{#1}\InBrackets{#2}}

\def\expect{\Ex}

\newcommand{\dd}{\mathrm{d}}  
\newcommand{\given}{\;\mid\;}

\newcommand{\rand}{\pi}
\newcommand{\ralloc}[1][\rand]{\alloc^{#1}}
\newcommand{\rprice}[1][\rand]{\price^{#1}}
\newcommand{\bid}{b}
\newcommand{\strat}{s}
\newcommand{\strats}{{\mathbf \strat}}

\newcommand{\virt}{\varphi}

\newcommand{\val}{v}
\newcommand{\vals}{{\mathbf \val}}
\newcommand{\vali}[1][i]{\val_{#1}}
\newcommand{\alloc}{x}

\newcommand{\price}{p}
\newcommand{\util}{u}

\newcommand{\capa}{C}
\newcommand{\ucapa}{\util_\capa}

\newcommand{\dist}{F}
\newcommand{\dens}{f}
\newcommand{\mech}{\mathcal M}
\DeclareMathOperator{\REV}{Rev}
\DeclareMathOperator{\WEL}{Welfare}

\newcommand{\qq}{\alloc}
\newcommand{\qc}{\qq_{\capa}}
\newcommand{\qv}{\qq_{\mathrm{val}}}

\newcommand{\lowval}{\val^{-}}
\newcommand{\highval}{\val^{+}}

\newcommand{\pricern}{\price^{\mathrm{RN}}}

\newcommand{\pricevc}{(\val - \capa)^+}
\newcommand{\virtplus}{\virt^+}

\newcommand{\pricebar}{\bar{\price}}

\newcommand{\qcbar}{\bar\qq_{\capa}}
\newcommand{\qvbar}{\bar\qq_{\mathrm{val}}}
\newcommand{\qqbar}{\bar{\qq}}
\newcommand{\crate}{r}

\newcommand{\pI}{\pricebar^{\mathrm{I}}}
\newcommand{\pII}{\pricebar^{\mathrm{II}}}
\newcommand{\pIII}{\pricebar^{\mathrm{III}}}

\newcommand{\valzero}{\val_0}

\newcommand{\bcap}{{\bid^{\capa}}}

\newcommand{\pricecap}{{\price^{\capa}}}
\newcommand{\dpricecap}{\tfrac{\dd}{\dd \val}\pricecap}
\newcommand{\dalloc}{\tfrac{\dd }{\dd \val}\alloc}

\newcommand{\ddright}{\dd_+}
\newcommand{\dpricecapright}{\tfrac{\ddright}{\ddright \val}\pricecap}

\newcommand{\pvc}{\price^{\mathrm{VC}}}
\newcommand{\dpvc}{\frac{\dd \pvc}{\dd \val}}

\newcommand{\dpricern}{\tfrac{\dd }{\dd \val}\pricern}
\newcommand{\utilreport}{U}
\newcommand{\urnreport}{\utilreport^{\mathrm{RN}}}
\newcommand{\ucapareport}{\utilreport^{\capa}}

\newcommand{\eps}{\epsilon}
\newcommand{\npart}{N}
\newcommand{\riemsum}{S}
\newcommand{\lowqc}{y}

\newcommand{\stoccom}[1]{#1}
\newcommand{\STOC}[1]{}
\newcommand{\NOTSTOC}[1]{#1}
\newcommand{\IFSTOCELSE}[2]{#2}

\title{Prior-independent Auctions for Risk-averse Agents}

\author{ Hu Fu
\and Jason Hartline
\and Darrell Hoy}

\begin{document}

\begin{titlepage}
\maketitle


We study simple and approximately optimal auctions for agents with a
particular form of risk-averse preferences.  We show that, for
symmetric agents, the optimal revenue (given a prior distribution over
the agent preferences) can be approximated by the first-price auction
(which is prior independent), and, for asymmetric agents, the optimal
revenue can be approximated by an auction with simple form.  These
results are based on two technical methods.  The first is for
upper-bounding the revenue from a risk-averse agent.  The second gives
a payment identity for mechanisms with pay-your-bid semantics.

\thispagestyle{empty}
\end{titlepage}

\newpage

\section{Introduction}
\label{sec:intro}

We study optimal and approximately optimal auctions for agents with
risk-averse preferences.  The economics literature on this subject is
largely focused on either comparative statics, i.e., is the
first-price or second-price auction better when agents are risk
averse, or deriving the optimal auction, e.g., using techniques from
optimal control, for specific distributions of agent preferences.  The
former says nothing about optimality but considers realistic
prior-independent auctions; the latter says nothing about realistic
and prior-independent auctions.  Our goal is to study approximately
optimal auctions for risk-averse agents that are realistic and not
dependent on assumptions on the specific form of the distribution of
agent preferences.  One of our main conclusions is that, while the
second-price auction can be very far from optimal for risk-averse
agents, the first-price auction is approximately optimal for an
interesting class of risk-averse preferences.

%
%
The microeconomic treatment of risk aversion in auction theory
suggests that the form of the optimal auction is very dependent on
precise modeling details of the preferences of agents, see, e.g.,
\citet{MR84} and \citet{M84}.  The resulting auctions are unrealistic
because of their reliance on the prior assumption and because they are
complex \citep[cf.][]{wil-87}.  Approximation can address both issues.
There may be a class of mechanisms that is simple, natural, and much
less dependent on exact properties of the distribution.  As an example
of this agenda for risk neutral agents, \citet{HR09} showed that for a
large class of distributional assumptions the second-price auction
with a reserve is a constant approximation to the optimal single-item
auction.  This implies that the only information about the
distribution of preferences that is necessary for a good approximation
is a single number, i.e., a good reserve price.  Often from this sort
of ``simple versus optimal'' result it is possible to do away with the
reserve price entirely.  \citet{DRY10} and \citet{RTY12} show that
simple and natural mechanisms are approximately optimal quite broadly.
We extend this agenda to auction theory for risk-averse agents.

%
%
The least controversial approach for modeling risk-averse agent
preferences is to assume agents are endowed with a concave function
that maps their wealth to a utility.  This introduces a non-linearity
into the incentive constraints of the agents which in most cases makes
auction design analytically intractable.  We therefore restrict
attention to a very specific form of risk aversion that is both
computationally and analytically tractable: utility functions that are
linear up to a given capacity and then flat. Importantly, an agent
with such a utility function will not trade off a higher probability
of winning for a lower price when the utility from such a lower price
is greater than her capacity.  While capacitated utility functions are
unrealistic, they form a basis for general concave utility functions.
In our analyses we will endow the benchmark optimal auction with
knowledge of the agents' value distribution and capacity; however,
some of the mechanisms we design to approximate this benchmark will be
oblivious to them.

%
%
As an illustrative example, consider the problem of maximizing welfare
by a single-item auction when agents have known capacitated utility
functions (but unknown values).  Recall that for risk-neutral agents
the second-price auction is welfare-optimal as the payments are
transfers from the agents to the mechanism and cancel from the
objective welfare which is thus equal to value of the winner.
(The auctioneer is assumed to have linear utility.)
For
agents with capacitated utility, the second-price auction can be far
from optimal.  For instance, when the difference between the highest
and second highest bid is much larger than the capacity then the
excess value (beyond the capacity) that is received by the winner does
not translate to extra utility because it is truncated at the
capacity.  Instead, a variant of the second-price auction, where the
highest bidder wins and is charged the maximum of 
the second highest bid and her bid less her capacity,
obtains the optimal welfare.
Unfortunately, this auction is parameterized by the form of the
utility function of the agents.  
There is, however, an auction, 
not dependent on specific knowledge of the utility functions or prior distribution, 
that is also welfare optimal:  If the agents values are
drawn i.i.d.\@ from a common prior distribution then the first-price
auction is welfare-optimal.  To see this: (a) standard analyses show
that at equilibrium the highest-valued agent wins, and (b) no agent
will shade her bid more than her capacity as she receives no
increased utility from such a lower payment but her probability of
winning strictly decreases.  

%
%
Our main goal is to duplicate the above observation for the objective
of revenue.  It is easy to see that the gap between the optimal
revenues for risk-neutral and capacitated agents can be of the same
order as the gap between the optimal welfare and the optimal revenue
(which can be unbounded).  When the capacities are small the revenue
of the welfare-optimal auction for capacitated utilities is close to
its welfare (the winners utility is at most her capacity).  Of course,
when capacities are infinite or very large then the risk-neutral
optimal revenue is close to the capacitated optimal revenue (the
capacities are not binding).  One of our main technical results shows
that even for mid-range capacities one of these two mechanisms that
are optimal at the extremes is close to optimal.

%
%
As a first step towards understanding profit maximization for
capacitated agents, we characterize the optimal auction for agents
with capacitated utility functions.  We then give a ``simple versus
optimal'' result showing that either the revenue-optimal auction
for risk-neutral agents or the above welfare-optimal auction for
capacitated agents is a good approximation to the revenue-optimal auction for
capacitated agents.  The Bulow-Klemperer \citeyearpar{BK96} Theorem
implies that with enough competition (and mild distributional
assumptions) welfare-optimal auctions are approximately
revenue-optimal.  Of course, the first-price auction is
welfare-optimal and prior-independent; therefore we conclude that it
is approximately revenue-optimal for capacitated agents.  

%
%
Our ``simple versus optimal'' result comes from an upper bound on the
expected payment of an agent in terms of her allocation rule
\citep[cf.][]{M81}.  This upper bound is the most technical result in
the paper; the difficulties that must be overcome by our analysis are
exemplified by the following observations.  First, unlike in
risk-neutral mechanism design, Bayes-Nash equilibrium does not imply
monotonicity of allocation rules.  There are mechanisms where an agent
with a high value would prefer less overall probabability of service
than she would have obtained if she had a lower value
(\autoref{ex:nonmono} in \autoref{sec:optimal}).  Second, even in the
case where the capacity is higher than the maximum value of any agent,
the optimal mechanism for risk-averse agents can generally obtain more
revenue than the optimal mechanism for risk-neutral agents
(\autoref{ex:v<C} in \autoref{sec:pricebound}).  This may be
surprising because, in such a case, the revenue-optimal mechanism for
risk-neutral agents would give any agent a wealth that is within the
linear part of her utility function.  Finally, while our upper bound
on risk-averse payments implies that this relative improvement is
bounded by a factor of two for large capacities, it can be arbitraily
large for small capacities (\autoref{ex:er-gap} in
\autoref{sec:pricebound}).

%
%
It is natural to conjecture that the first-price auction will continue
to perform nearly optimally well beyond our simple model (capacitated
utility) of risk-averse preferences.  It is a relatively
straightforward calculation to see that for a large class of
risk-averse utility functions from the literature \citep[e.g.,][]{M84}
the first-price auction is approximately optimal at extremal risk
parameters (risk-neutral or extremely risk-averse).  
We leave to future work the
extension of our analysis to mid-range risk parameters for these other
families of risk-averse utility functions.

%
%
It is significant and deliberate that our main theorem is about the
first-price auction which is well known to not have a truthtelling
equilibrium.  Our goal is a prior-independent mechanism.  In
particular, we would like our mechanism to be parameterized neither by the
distribution on agent preference nor by the capacity that governs the
agents utility function.  While it is standard in mechanism design and
analysis to invoke the {\em revelation principle} \citep[cf.][]{M81}
and restrict attention to auctions with truthtelling as equilibrium,
this principle cannot be applied in prior-independent auction design.
An auction with good equilibrium can be implemented by one with 
truthtelling as an equilibrium if the agent strategies can be
simulated by the auction.  In a Bayesian environment, agent strategies
are parameterized by the prior distribution and therefore the
suggested revelation mechanism is not generally prior independent.

\paragraph{Risk Aversion, Universal Truthfulness, and Truthfulness in Expectation\stoccom{.}}
Our results have an important implication on a prevailing and
questionable perspective that is explicit and implicit broadly in the
field of algorithmic mechanism design.  Two standard solution concepts
from algorithmic mechanism design are ``universal truthfulness'' and
``truthfulness in expectation.''  A mechanism is universally truthful
if an agent's optimal (and dominant) strategy is to reveal her values
for the various outcomes of the mechanism regardless of the reports of
other agents or random coins flipped by the mechanism.  In contrast,
in a truthful-in-expectation mechanism, revealing truthfully her values
only maximizes the agent's utility in expectation over the random
coins tossed by the mechanism.  Therefore, a risk-averse agent modeled
by a non-linear utility function may not bid truthfully in a
truthful-in-expectation mechanism designed for risk-neutral agents,
whereas in a universally truthful mechanism an agent behaves the same
regardless of her risk attitude.  For this reason, the above-mentioned
perspective sees universally truthful mechanisms superior because the
performance guarantees shown for risk-neutral agents seem to apply to
risk-averse agents as well.

This perspective is incorrect because the optimal performance possible
by a mechanism is different for risk-neutral and risk-averse agents.
In some cases, a mechanism may exploit the risk attitude of the agents
to achieve objectives better than the optimal possible for
risk-neutral agents; 
in other cases, the objective itself relies on the utility functions
(e.g.\@ social welfare maximization), and therefore the same outcome
has a different objective value.  In all these situations, the
performance guarantee of universally truthful mechanisms measured by
the risk-neutral optimality loses its meaning.  We have already
discussed above two examples for capacitated agents that illustrate
this point: for welfare maximization the second-price auction is not
optimal, for revenue maximization the risk-neutral revenue-optimal
auction can be far from optimal.

The conclusion of the discussion above is that the universally
truthful mechanisms from the literature are not generally good when
agents are risk averse; therefore, the solution concept of universal
truthfulness buys no additional guarantees over truthfulness in
expectation.  Nonetheless, our results suggest that it may be possible
to develop a general theory for prior-independent mechanisms for
risk-averse agents.  By necessity, though, this theory will look
different from the existing theory of algorithmic mechanism design.

\paragraph{Summary of Results\stoccom{.}}

Our main theorem is that the first-price auction is a
prior-independent $5$-approximation for revenue for two or more
agents with i.i.d.\@ values and risk-averse preferences (given by a
common capacity).  The technical results that enable this theorem are
as follows:

\begin{itemize}
\item The optimal auction for agents with capacitated utilities is a
  two-priced mechanism where a winning agent either pays her full
  value or her value less her capacity.

\item The expected revenue of an agent with capacitated utility and
  regular value distribution can be bounded in terms of an expected
  (risk-averse) virtual surplus, where the (risk-averse) virtual value
  is twice the risk-neutral virtual value plus the value minus
  capacity (if positive).

\item Either the mechanism that optimizes value minus capacity (and
  charges the Clarke payments or value minus capacity, whichever is
  higher) or the risk-neutral revenue optimal mechanism is a
  3-approximation to the revenue optimal auction for capacitated
  utilities.  

\item We characterize the Bayes-Nash equilibria of auctions with
  capacitated agents where each bidder's payment when served is a
  deterministic function of her value.  An example of this is the
  first-price auction.
The BNE strategies of the capacitated agents can be calculated
formulaically from the BNE strategies of risk-neutral agents.
\end{itemize}

Some of these results extend beyond single-item auctions.  In
particular, the characterization of equilibrium in the first-price
auction holds for position auction environments (i.e., where agents
are greedily by bid assigned to positions with decreasing
probabilities of service and charged their bid if served).  Our
simple-versus-optimal 3-approximation holds generally for
downward-closed environments, non-identical distributions, and
non-identical capacities.

\paragraph{Related Work\stoccom{.}}
The comparative performance of first- and second-price auctions 
in the presence of risk aversion has been well studied in the 
Economics literature. From a revenue perspective, 
first-price auctions are shown to outperform second-price auctions very 
broadly. \citet{RS81} and \citet{H80} show this for symmetric settings where bidders have the same concave utility
function. \citet{MR84} show this for more general preferences.

\citet{M87} shows that in addition to the revenue dominance, bidders whose 
risk attitudes exhibit \emph{constant absolute risk aversion (CARA)} are indifferent 
between first- and second-price auctions, even though they pay more in 
expectation in the first-price auction. \citet{HMZ10} considers the optimal 
reserve prices to set in each, and shows that the optimal reserve in the 
first price auction is less than that in the second price auction.  Interestingly, under 
light conditions on the utility functions, as risk aversion increases, the 
optimal first-price reserve price decreases.

\citet{Matthews83} and \citet{MR84} have considered optimal mechanisms for a single item, with symmetric bidders
(i.i.d.\@ values and identical utility function), for CARA and more general preferences.

Recently, \citet{DP12} have shown that by insuring bidders against
uncertainty, any truthful-in-expectation mechanism for risk-neutral
agents can be converted into a dominant-strategy incentive compatible
mechanism for risk-averse buyers with no loss of revenue. However,
there is potentially much to gain---mechanisms for risk-averse buyers
can achieve unboundedly more welfare and revenue than mechanisms for
risk-neutral bidders, as we show in \autoref{ex:er-gap} of
\autoref{sec:pricebound}.




\section{Preliminaries}
\label{sec:prelim}
\paragraph{Risk-averse Agents\stoccom{.}}  

Consider selling an item to an agent who has a private
valuation~$\val$ drawn from a known distribution~$\dist$.  Denote the
outcome by $(\alloc, \price)$, where $\alloc \in \{0, 1\}$ indicates
whether the agent gets the item, and $\price$ is the payment made.
The agent obtains a {\em wealth} of $\val\alloc - \price$ for such an
outcome and the agent's utility is given by a concave utility function
$\util(\cdot)$ that maps her wealth to utility, i.e., her utility for
outcome $(\alloc,\price)$ is $\util(\val\alloc - \price)$.  Concave
utility functions are a standard approach for modeling
risk-aversion.\footnote{There are other definitions of risk aversion;
  this one is the least controversial.  See \citet{MWG95} for a thorough exposition of expected utility theory.}

A {\em capacitated} utility function is $\ucapa(z) = \min (z,\capa)$
for a given $\capa$ which we refer to as the {\em capacity}.
Intuitively, small $\capa$ corresponds to severe risk aversion; large
$\capa$ corresponds to mild risk aversion; and $C = \infty$
corresponds to risk neutrality.  An agent views an auction as a
deterministic rule that maps a random source and the (possibly random)
reports of other agents which we summarize by $\rand$, and the report
$\bid$ of the agent, to an allocation and payment.  We denote these
coupled allocation and payment rules as $\ralloc(\bid)$ and
$\rprice(\bid)$, respectively.  The agent wishes to maximize her
expected utility which is given by $\expect[\rand]{\ucapa(\val
  \ralloc(\bid) - \rprice(\bid))}$, i.e., she is a von
Neumann-Morgenstern utility maximizer.

\paragraph{Incentives\stoccom{.}}
A strategy profile of agents is $\strats = (\strat_1,\ldots,\strat_n)$
mapping values to reports.  Such a strategy profile is in {\em
  Bayes-Nash equilibrium} (BNE) if each
agent~$i$ maximizes her utility by reporting $\strat_i(\val_i)$.
I.e., for all $i$, $\val_i$, and~$z$:
$$
\expect[\rand]{\util(\val_i \ralloc_i(\strat_i(\val_i)) - \rprice_i(\strat_i(\val_i)))}
 \geq 
\expect[\rand]{\util(\val_i \ralloc_i(z) - \rprice_i(z))}
$$ where $\rand$ denotes the random bits accessed by the mechanism as
well as the random inputs $\strat_j(\val_j)$ for $j \neq i$ and
$\val_j \sim \dist_j$.  A mechanism is {\em Bayesian incentive
  compatible} (BIC) if truthtelling is a Bayes-Nash equilibrium:
for all $i$, $\val_i$, and $z$
\begin{align}
\label{eq:IC}
\expect[\rand]{\util(\val_i \ralloc_i(\val_i) - \rprice_i(\val_i))}
 \geq 
\expect[\rand]{\util(\val_i \ralloc_i(z) - \rprice_i(z))}
\tag{IC}
\end{align} where $\rand$ denotes the random bits accessed by the mechanism as
well as the random inputs $\val_j \sim \dist_j$ for $j \neq i$.  

We will consider only mechanisms where losers have no payments, and
winners pay at most their bids.  These constraints imply ex post {\em
  individual rationality} (IR).  Formulaically, for all $i$, $\vali$, and
$\pi$, $\rprice_i(\val_i) \leq \vali$ when $\ralloc_i(\val_i) = 1$ and
$\rprice_i(\val_i) = 0$ when $\ralloc_i(\val_i) = 0$.




%


\paragraph{Auctions and Objectives\stoccom{.}}

%
%
The revenue of an auction $\mech$ is the total payment of all agents;
its expected revenue for implicit distribution $\dist$ and Bayes-Nash
equilibrium is denoted $\REV(\mech) = \expect[\rand,\vals]{\sum_i
  \rprice_i(\val_i)}$.
The welfare
of an auction $\mech$ is the total utility of all participants
including the auctioneer; its expected welfare is denoted $\WEL(\mech)
= \REV(\mech) + \expect[\rand,\vals]{\sum_i \util(\val_i
  \ralloc_i(\val_i) - \rprice_i(\val_i))}$.

%
%

%
%
Some examples of auctions are:
the {\em first-price auction} (FPA) serves the agent with the highest bid and
charges her her bid; the {\em second-price auction} (SPA) serves the
agent with the highest bid and charges her the second-highest bid.
The second price auction is incentive compatible regardless of agents'
risk attitudes.  The {\em capacitated second-price auction} (CSP)
serves the agent with the highest bid and charges her the maximum of
her value less her capacity and the second highest bid.  The
second-price auction for capacitated agents is incentive compatible
for capacitated agents because, relative to the second-price auction,
the utility an agent receives for truthtelling is unaffected and the
utility she receives for any misreport is only (weakly) lower.

\paragraph{Two-Priced Auctions\stoccom{.}}
The following class of auctions will be relevant for agents with
capacitated utility functions.
\begin{definition}
\label{def:two-price}
A mechanism~$\mech$ is \emph{two-priced} if, whenever $\mech$ serves
an agent with capacity $\capa$ and value $\val$, the agent's payment
is either $\val$ or $\val-\capa$; and otherwise (when not served) her
payment is zero.  Denote by $\qv(\val)$ and $\qc(\val)$  probability of
paying $\val$ and $\val - \capa$, respectively.
\end{definition}
\noindent
Note that from an agent's perspective the outcome of a two-priced
mechanism is fully described by a $\qc$ and $\qv$.



\paragraph{Auction Theory for Risk-neutral Agents\stoccom{.}}

For risk neutral agents, i.e., with $\util(\cdot)$ equal to the
identity function, only the probability of winning and expected
payment are relevant.  The {\em interim allocation rule} and {\em
  interim payment rule} are given by the expectation of $\ralloc$ and
$\rprice$ over~$\rand$ and denoted as $\alloc(\bid) =
\Ex[\rand]{\ralloc(\bid)}$ and $\price(\bid) =
\Ex[\rand]{\rprice(\bid)}$, respectively (recall that $\rand$ encodes
the randomization of the mechanism and the reports of other agents).

For risk-neutral agents, \citet{M81} characterized interim allocation
and payment rules that arise in BNE and solved for the revenue optimal
auction.  
These results are summarized in the following theorem.
\begin{theorem}[\citealp{M81}]
\label{thm:myerson}
For risk neutral bidders with valuations drawn independently and
identically from $\dist$,
\begin{enumerate}
\item (monotonicity) 
\label{thmpart:monotone}
The allocation rule $\alloc(\val)$ for each agent is monotone
  non-decreasing in $\val$.
\item 
\label{thmpart:payment}
(payment identity) The payment rule satisfies $\price(\val) = \val \alloc(\val) - \int_0^\val\alloc(z) \dd z$.
\item 
\label{thmpart:virt}
(virtual value) The ex ante expected payment of an agent is
  $\expect[\val]{\price(\val)} = \expect[\val]{\virt(\val)\alloc(\val)}$ where
  $\virt(\val) = \val - \frac{1-\dist(\val)}{\dens(\val)}$ is the {\em virtual value} for value $\val$.
\item
\label{thmpart:opt}
(optimality) When the distribution $\dist$ is {\em regular}, i.e., $\virt(\val)$ is
monotone, the second-price auction with reserve $\virt^{-1}(0)$ is
revenue-optimal.
\end{enumerate}
\end{theorem}
\noindent The payment identity in \autoref{thmpart:payment} implies
the {\em revenue equivalence} between any two auctions with the same
BNE allocation rule.


A well-known result by \citeauthor{BK96} shows that, in \autoref{thmpart:opt} of \autoref{thm:myerson}, instead of
having a reserve price to make the second-price auction optimal, one may as well add in another identical bidder to get at least as much revenue.

\begin{theorem}[\citealp{BK96}]
\label{thm:bk}
For risk neutral bidders with valuations drawn i.i.d.\@ from a regular distribution, the revenue from the second-price auction with $n+1$
bidders is at least that of the optimal auction for $n$~bidders.
\end{theorem}

\section{The Optimal Auctions}
\label{sec:optimal}
In this section we study the form of optimal mechanisms for
capacitated agents.  In \autoref{sec:optimal-two-price}, we show that
it is without loss of generality to consider two-priced auctions, and
in \autoref{sec:two-price-BIC} we characterize the incentive
constraints of two-priced auctions.  In \autoref{sec:polytime} we use
this characterization to show that the optimal auction (in discrete
type spaces) can be computed in polynomial time in the number of
types.

\subsection{Two-priced Auctions Are Optimal}
\label{sec:optimal-two-price}

Recall a two-priced auction is one where when any agent is served she is
either charged her value or her value minus her capacity.  We show
below that restricting our attention to two-priced auctions is without
loss for the objective of revenue.


\begin{theorem}
\label{thm:optimal-two-price}
For any auction on capacitated agents there is a two-priced auction
with no lower revenue.
\end{theorem}

\begin{proof}
We prove this theorem in two steps.  In the first step we show, quite
simply, that if an agent with a particular value received more wealth
than $\capa$ then we can truncate her wealth to $\capa$ (by charging
her more).  With her given value she is indifferent to this change,
and for all other values this change makes misreporting this value
(weakly) less desirable.  Therefore, such a change would not induce
misreporting and only (weakly) increases revenue.  This first step
gives a mechanism wherein every agent's wealth is in the linear part
of her utility function.  The second step is to show that we can
transform the distribution of wealth into a two point distribution.
Whenever an agent with value~$\val$ is offered a price that results in a wealth $w \in [0, \capa]$, we instead offer her a price
of $\val - \capa$ with probability $w / \capa$, and a price of $\val$ with the remaining probability.  Both the expected
revenue and the utility of a truthful bidder is unchanged.  The expected utility of other types to misreport~$\val$, however,
weakly decreases by the concavity of $\ucapa$, because mixing over endpoints of an interval on a concave function gives
less value than mixing over internal points with the same expectation.
\end{proof}

\subsection{Characterization of Two-Priced Auctions}
\label{sec:two-price-BIC}

In this section we characterize the incentive constraints of
two-priced auctions.  We focus on the induced two-priced mechanism for a single agent given
the randomization $\rand$ of other agent values and the mechanism.
The interim two-priced allocation rule of this agent is denoted by
$\qq(\val) = \qv(\val) + \qc(\val)$.

\begin{lemma} 
\label{lem:two-price-BIC}
A mechanism with two-price allocation rule $\qq  = \qv + \qc$ is BIC if
and only if for all $\val$ and $\highval$ such that $\val < \highval
\leq \val + \capa$,
\begin{equation}
\label{eq:near-bic}
\frac{\qv(\val)}{\capa} \leq \frac{ \qc(\highval) - \qc(\val)}{\highval - \val} \leq \frac{\qq(\highval)}{\capa}.
\end{equation}
\end{lemma}


Equation \eqref{eq:near-bic} can be equivalently written as the
following two linear constraints on $\qc$, for all $\lowval \leq \val \leq
\highval \in [\val -\capa,\val + \capa]$:
\begin{align}
\qc(\highval) &\geq \qc(\val) + \frac{\highval-\val}{\capa}\cdot \qv(\val), \label{eq:bic-underbidding}\\ 
\qc(\lowval) &\geq \qc(\val) - \frac{\val-\lowval}{\capa}\cdot \qq(\val).\label{eq:bic-overbidding}
\end{align}
Equations \eqref{eq:bic-underbidding} and \eqref{eq:bic-overbidding}
are illustrated in \autoref{fig:opt-bic}.  For a fixed $\val$,
\eqref{eq:bic-underbidding} with $\highval=\val+\capa$ yields a lower
bounding line segment from $(\val, \qc(\val))$ to $(\val + \capa,
\qc(\val) + \qv(\val))$, and \eqref{eq:bic-overbidding} with $\lowval
= \val - \capa$ gives a lower bounding line segment from $(\val,
\qc(\val))$ to $(\val - \capa, \qc(\val)-\qq(\val))$.  Note that
\eqref{eq:bic-underbidding} implies that $\qc$ is monotone.

In the special case when $\qc$ is differentiable, by taking $\highval$
approaching $\val$ in \eqref{eq:near-bic}, we have
$\tfrac{\qv(\val)}{\capa} \leq \qc'(\val) \leq
\tfrac{\qq(\val)}{\capa}$ for all~$\val$.  In general, we have the
following condition in the integral form (see \autoref{sec:opt-app}
for a proof).
\begin{corollary}
\label{cor:int-bic}
The allocation rule $\qq = \qv + \qc$ of a BIC two-priced mechanism
for all $\val < \highval$ satisfies:
\begin{align}
\label{eq:int-bic}
 \int_{\val}^{\highval} \frac{\qv(z)}{\capa} \: \dd z \leq \qc(\highval) - \qc(\val) \leq \int_{\val}^{\highval}
\frac{\qq(z)}{\capa} \: \dd z. 
\end{align}
\end{corollary}


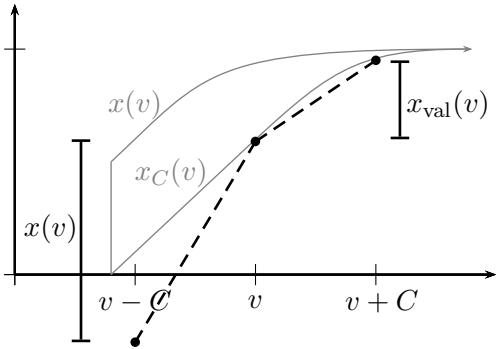
\begin{figure}[!ht]

\begin{center}
\begin{pspicture}(-1, -1)(6,4)
	\psset{yunit=3cm,xunit=1.6cm}
	\psaxes[Dy=1, Dx=1, labels=none]{->}(0, 0)(0,0)(4, 1.2)%
{\psset{linecolor=gray}%
	\psline[linewidth=0.5pt, linearc=0]{-}(0.8,0)(0.8, 0.5)
	\psbezier[linewidth=0.5pt]{->}(0.8,0.5)(1.6,0.9)(1.6,1)(3.8,1)
	\rput(1, 0.75){\gray $\qq(\val)$}
	\psline[linewidth=0.5pt]{-}(.8,0)(1.8, 0.5)
	\psbezier[linewidth=0.5pt]{->}(1.8,0.5)(2.6,0.9)(2.8,1)(3.8,1) 
}%
	\rput(1.3,0.45){\gray $\qc(\val)$}
	\rput(1,-0.12){$\val - \capa$}
	\rput(2,-0.12){$\val$}
	\rput(3.05,-0.12){$\val + \capa$}	
	
	\psline[linewidth=1pt, linestyle=dashed]{-}(2, 0.59)(3,0.95)
	\psline[linewidth=1pt, linestyle=dashed]{-}(1,-0.3)(2, 0.59)
	\psdots*(1,-0.3)(2, 0.59)(3,0.95)
		
	\psline[linewidth=1pt]{|-|}(3.2,0.6)(3.2, 0.95)
	\rput(3.6,0.75){$\qv(\val)$}
	\psline[linewidth=1pt]{|-|}(0.55, -0.3)(0.55,0.6)
	\rput(.30,0.2){$\qq(\val)$}

\end{pspicture}
\end{center}
\caption{Fixing $\qq(\val) = \qv(\val) + \qc(\val)$, the dashed
  line between points $(\val - \capa, \qc(\val)-\qq(\val))$, $(\val,
  \qc(\val))$, and $(\val + \capa, \qc(\val) + \qv(\val))$ (denoted by
  ``$\bullet$'') depicts the lower bounds from
  \eqref{eq:bic-underbidding} and \eqref{eq:bic-overbidding} on $\qc$
  for values in $[\val - \capa,\val+\capa]$.\label{fig:opt-bic}}
\end{figure}


Importantly, the equilbrium characterization of two-priced mechanisms
does not imply monotonicity of the allocation rule $\alloc$.  This is
in contrast with mechanisms for risk-neutral agents, where incentive
compatibility requires a monotone allocation rule
(\autoref{thm:myerson}, \autoref{thmpart:monotone}).  This
non-monotonicity is exhibited in the following example.

\begin{example}
\label{ex:nonmono}
There is a single-agent two-priced mechanism with a non-monotone
allocation rule.  Our agent has two possible values $\val = 3$ and
$\val = 4$, and capacity $\capa$ of $2$.  We give a two price
mechanism.  Recall that $\qc(\val)$ is the probability with which the
mechanism sells the item and charges $\val - \capa$; $\qv(\val)$ is
the probability with which the mechanism sells the item and charges
$\val$; and $\qq(\val) = \qc(\val) + \qv(\val)$.  The mechanism and its outcome are summarized in the following table.

\begin{center}
\begin{tabular}{|c|c|c|c|c|c|}
\hline
\val & $\qq$ & $\qc$ & $\qv$ & \shortstack{utility from \\truthful reporting} & \shortstack{utility from\\ misreporting} \\ \hline
3 & 5/6 & 1/2 & 1/3  & 1 & 2/3 \\ \hline
4 & 2/3 & 2/3 & 0 & 4/3 & 4/3 \\ \hline
\end{tabular}
\end{center}
\end{example}

\subsection{Optimal Auction Computation}
\label{sec:polytime}

Solving for the optimal mechanism is computationally tractable for any
discrete (explicitly given) type space $T$.  Given a discrete
valuation distribution on support~$T$, one can use $2\left|T\right|$ variables to
represent the allocation rule of any two-priced mechanism, and the
expected revenue is a linear sum of these variables.
\autoref{lem:two-price-BIC} shows that one can use $O(|T|^2)$ linear
constraints to express all BIC allocations, and hence the revenue
optimization for a single bidder can be solved by a $O(|T|^2)$-sized
linear program.  Furthermore, using techniques developed by
\citet{CDW12} and \citet{AFHHM12}, in particular the ``token-passing''
characterization of single-item auctions by \citet{AFHHM12}, we
obtain:

\begin{theorem}
\label{thm:polytime}
For $n$~bidders with independent valuations with type spaces $T_1, \cdots, T_n$ and capacities $\capa_1, \cdots,
\capa_n$, one can solve for the optimal single-item auction with a linear program of size $O\left(\left(\sum_i |T_i| \right)^2\right)$.
\end{theorem}

\section{An Upper Bound on Two-Priced Expected Payment}
\label{sec:pricebound}

In this section we will prove an upper bound on the expected payment
from any capacitated agent in a two-priced mechanism.  This upper bound
is analogous in purpose to the identity between expected risk-neutral
payments and expected virtual surplus of \citet{M81} from which
optimal auctions for risk-neutral agents are derived.  We use this
bound in \autoref{sec:firstprice} and \autoref{sec:one-v-two} to
derive approximately optimal mechanisms.

As before, we focus on the induced two-priced mechanism for a single agent given
the randomization $\rand$ of other agent values and the mechanism.
The expected payment of a bidder of value~$\val$ under allocation rule $\qq(\val) = \qc(\val) + \qv(\val)$
is $\price(\val) = \val\cdot \qv(\val) + (\val-\capa)\cdot \qc(\val) =
\val \cdot \qq(\val) - \capa \cdot \qc(\val)$.

Recall from \autoref{thm:myerson} that the (risk-neutral) virtual
value for an agent with value drawn from distribution $\dist$ is
$\virt(\val) = \val - \frac{1-\dist(\val)}{\dens(\val)}$ and that the
expected risk-neutral payment for allocation rule $\alloc(\cdot)$ is
$\Ex[\val]{\virt(\val)\alloc(\val)}$.  Denote $\max(0,\virt(\val))$ by
$\virtplus(\val)$ and $\max(\val - \capa,0)$ by $\pricevc$.




\begin{theorem}
\label{thm:pricebound}
For any agent with value $\val \sim \dist$, capacity~$\capa$, and two-priced allocation rule $\qq(\val) = \qv(\val) + \qc(\val)$, 
\begin{align*}
\Ex[\val]{\price(\val)} \leq &
\Ex[\val]{\virtplus(\val) \cdot \qq(\val)}
+ \Ex[\val]{\virtplus(\val) \cdot \qc(\val)}  \\
& + \Ex[\val]{\pricevc \cdot \qc(\val)}.
\end{align*}
\end{theorem}


\begin{corollary}
\label{cor:revbound}
When bidders have regular distributions and a common capacity, either
the risk-neutral optimal auction or the capacitated second price
auction (whichever has higher revenue) gives a 3-approximation to the
optimal revenue for capacitated agents.
\end{corollary}

\begin{proof}
For each of the three parts of the revenue upper bound of
\autoref{thm:pricebound}, there is a simple auction that optimizes the
expectation of the part across all agents.  For the first two parts,
the allocation rules across agents (both for $\qq(\cdot)$ and
$\qc(\cdot)$) are feasible.  When the distributions of agent values
are regular (i.e., the virtual value functions are monotone), the
risk-neutral revenue-optimal auction optimizes virtual surplus across
all feasible allocations (i.e., expected virtual value of the agent
served); therefore, its expected revenue upper bounds the first and
second parts of the bound in \autoref{thm:pricebound}.  The revenue of
the third part is again the expectation of a monotone function (in
this case $\pricevc$) times the service probability.  The auction that
serves the agent with the highest (positive) ``value minus capacity''
(and charges the winner the maximum of her ``minimum winning bid,'' i.e., the second-price payment rule, and her ``value minus
capacity'') optimizes such an expression over all feasible
allocations; therefore, its revenue upper bounds this third part of
the bound in \autoref{thm:pricebound}.  When capacities are identical,
this auction is the capacitated second price auction.
\end{proof}


Before proving \autoref{thm:pricebound}, we give two examples.  The
first shows that the gap between the revenue of the capacitated
second-price auction and the risk-neutral revenue-optimal auction
(i.e., the two auctions from \autoref{cor:revbound}) can be
arbitrarily large.  This means that there is no hope that an auction
for risk-neutral agents always obtains a good revenue for risk-averse
agents.  The second example shows that even when all values are
bounded from above by the capacity (and therefore, capacities are
never binding in a risk-neutral auction) an auction for risk-averse
agents can still take advantage of risk aversion to generate higher
revenue.  Consequently, the fact that we have two risk-neutral revenue
terms in the bound of \autoref{thm:pricebound} is necessary (as the
``value minus capacity'' term is zero in this case).

\begin{example}
\label{ex:er-gap}
The {\em equal revenue} distribution on interval $[1,h]$ has
distribution function $\dist(z) = 1-1/z$ (with a point mass at $h$).
The distribution gets its name because such an agent would accept
any offer price of $\price$ with probability $1/\price$ and generate
an expected revenue of one.  With one such agent the optimal
risk-neutral revenue is one.  Of course, an agent with capacity $\capa
= 1$ would happily pay her value minus her capacity to win all the
time (i.e., $\qq(\val) = \qc(\val) = 1$).  The revenue of this auction
is $\Ex{\val} - 1 = \ln h$.  For large~$h$, this is unboundedly larger
than the revenue we can obtain from a risk-neutral agent with the
same distribution.
\end{example}

\begin{example}
\label{ex:v<C}
The revenue from a two-priced mechanism can be better than the optimal
risk-neutral revenue even when all values are no more than the
capacity. Consider selling to an agent with capacity of $\capa = 1000$ and
value drawn from the equal revenue distribution from
\autoref{ex:er-gap} with $h=1000$.

The following two-priced rule is BIC and generates revenue of
approximately $1.55$ when selling to such a bidder. Let $\qc(\val) =
\frac{0.6}{1000}(\val-1)$, $\qq(\val) = \min(\qc(\val) + 0.6, 1)$, and
$\qv(\val) = \qq(\val) - \qc(\val)$ (shown in
\autoref{fig:ex-eqrevunderC}). Recall that the expected payment from
an agent with value $\val$ can be written as $\val\qq(\val) - \capa
\qc(\val)$; for small values, this will be approximately $0.6$; for
large values this will increase to $400$. The expected revenue is
$\int_1^{1000} \left( z \cdot \qq(z) -1000\ \qc(z) \right) \dens (z)
\dd z + \frac{1}{1000}(1000 \cdot \qv(1000)) \approx 1.15 + 0.4
\approx 1.55$, an improvement over the optimal risk-neutral revenue of
$1$.
\end{example}

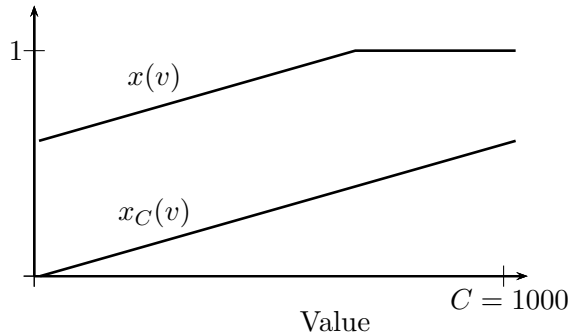
\begin{figure}[!hb]
			\centering
		\begin{pspicture}(-0.5,-1)(6.5,4)
			\psset{yunit=3cm,xunit=1.6cm}
			\psaxes[Dy=1, Dx=3.9, labels=none]{->}(0, 0)(0,0)(4.1, 1.2)
			\psline[linewidth=1pt]{-}(0.04,0.6)(2.668, 1)(4, 1)
			\psline[linewidth=1pt]{-}(0.05, 0)(4, 0.6)
			
			\rput(1, 0.88){$\qq(\val)$}
			\rput(1, 0.28){$\qc(\val)$}
			\rput(3.95, -.1){$C=1000$}
			\rput(-0.15, 1){$1$}
			\rput(2.5,-0.2){Value}
		\end{pspicture}		
	\caption{With $\capa=1000$ and values from the equal revenue
          distribution on $[1,1000]$, this two-priced mechanism is BIC
          and achieves $1.55$ times the revenue of the optimal
          risk-neutral mechanism.\label{fig:ex-eqrevunderC}}
\end{figure}

In the remainder of this section we instantiate the following outline
for the proof of \autoref{thm:pricebound}.  First, we transform any
given two-priced allocation rule $\qq = \qv + \qc$ into a new two-priced rule
$\qqbar(\val) = \qcbar(\val) + \qvbar(\val)$ (for which the
expected payment is $\pricebar(\val) = \val \qqbar(\val) - \capa
\qcbar(\val)$).  While this transformation may violate some incentive
constraints (from \autoref{lem:two-price-BIC}), it enforces convexity
of $\qcbar(\val)$ on $\val \in [0,\capa]$ and (weakly) improves
revenue.  Second, we derive a simple upper bound on the payment rule
$\pricebar(\cdot)$.  Finally, we use the enforced convexity property
of $\qcbar(\cdot)$ and the revenue upper bound to partition the
expected payment $\Ex[\val]{\pricebar(\val)}$ by the three terms that
can each be attained by simple mechanisms.

\subsection{Two-Priced Allocation Construction}
\label{sec:mechbar}

We now construct a two-priced allocation rule $\qqbar = \qvbar + \qcbar$
from $\qq = \qv + \qc$ for which (a) revenue is improved, i.e.,
$\pricebar(\val) \geq \price(\val)$, and (b) the probability the agent
pays her value minus capacity, $\qcbar(\val)$, is convex for $\val \in
[0,\capa]$.  In fact, given $\qv$, $\qcbar$ is the smallest function
for which IC constraint \eqref{eq:bic-underbidding} holds; and in the
special case when $\qv$ is monotone, the left-hand side of
\eqref{eq:int-bic} is tight for $\qcbar$ on $[0, \capa]$.  Other
incentive constraints may be violated by $\qqbar$, but we use it only
as an upper bound for revenue.

\begin{definition}[$\qqbar$]
\label{def:mechbar}
We define $\qqbar = \qcbar + \qvbar$ as follows:
\begin{enumerate}
\item $\qvbar(\val) = \qv(\val)$;
\item Let $\crate(\val)$ be $\tfrac{1}{\capa} \sup_{z \leq \val} \qv(z)$, and let
\begin{align}
\label{eq:qcbar}
\qcbar(\val) = \left\{ \begin{array}{ll}
\int_0^{\val} \crate(y) \: \dd y, & \val \in [0, \capa]; \\
\qc(\val), & \val > \capa.
\end{array}
\right.
\end{align}
\end{enumerate}
\end{definition}

\begin{lemma}[Properties of $\qqbar$]
\label{thm:mechbar}
\ \\
\begin{enumerate}
\item \label{prop:qcbar-conv} On $\val \in [0,\capa]$, $\qcbar(\cdot)$
  is a convex, monotone increasing function.

\item \label{lem:qcbar}
On all $\val$, $\qcbar(\val) \leq \qc(\val)$.

\item 
\label{thmpart:qcbar-bic}
The incentive
constraint from the left-hand side of \eqref{eq:int-bic} holds for
$\qcbar$: 
$\frac{1}{\capa}\int_\val^{\highval} \qvbar(z)\: \dd z \leq \qcbar(\highval) -
\qcbar(\val)$ for all $\val < \highval$.

\item \label{cor:qcbar} On all $\val$, $\qcbar(\val) \leq \qc(\val)$,
  $\qqbar(\val) \leq \qq(\val)$, and $\pricebar(\val) \geq
  \price(\val)$.
\end{enumerate}
\end{lemma}

%
The proof of \autoref{lem:qcbar} is technical, and we give a sketch
here.  Recall that, for each $\val$, the IC constraint
\eqref{eq:bic-underbidding} gives a linear constraint lower
bounding~$\qc(\highval)$ for every $\highval > \val$.  If one
decreases $\qc(\val)$, the lower bound it imposes on~$\qc(\highval)$
is simply ``pulled down'' and is less binding.  The definition
of~$\qcbar$ simply lands $\qcbar(\val)$ on the most binding lower
bound, and therefore not only makes $\qcbar(\val)$ at most
$\qc(\val)$, but also lowers the linear constraint that $\val$ imposes
on larger values.  If the number of values is countable or if $\qv$ is
piecewise constant, the lemma is easy to see by induction.  A full
proof for the general case of \autoref{lem:qcbar}, along with the
proofs of the other more direct parts of \autoref{thm:mechbar}, is
given in \autoref{sec:pricebound-app}.

\subsection{Payment Upper Bound}
\label{sec:pricebar}


	\begin{figure}[!hb]
			\small
			\centering
			\hspace*{\fill}
			\subfloat[][Shaded region is the expected payment from an agent of value $\val$. \label{fig:paymentblockregion} ]{
				\begin{pspicture}(-1, -1)(6,4)
				\psset{yunit=3cm,xunit=1.6cm}
				\psaxes[Dy=1, Dx=1, labels=none, ticks=y]{->}(0, 0)(0,0)(4, 1.2)
				\pspolygon[fillstyle=solid, fillcolor=lightgray](0,0)(0,0.93)(2,0.93)(2,0.6)(1,0.6)(1,0)				
				\psline[linewidth=1pt, linearc=0]{-}(0.8,0)(0.8, 0.5)
				\psbezier[linewidth=1pt]{->}(0.8,0.5)(1.6,0.9)(1.6,1)(3.8,1)
				\rput(2.5, 1.1){$\qqbar(\val)$}
				\psline[linewidth=1pt]{-}(.8,0)(1.8, 0.5)
				\psbezier[linewidth=1pt]{->}(1.8,0.5)(2.6,0.9)(2.8,1)(3.8,1) 
				\rput(2.5,0.6){$\qcbar(\val)$}
				\rput(1,-0.1){$\val - \capa$}
				\rput(2,-0.1){$\val$}
				\rput(3.5,-0.2){Value}
				\psline(0, -0.04)(0, 0.04)
				\psline(1, -0.04)(1, 0.04)
				\psline(2, -0.04)(2, 0.04)
		\end{pspicture}
		}\hfill
			\subfloat[][Shaded region upper bounds expected payment from an agent with value $\val$, shown in
\autoref{lem:pricebar}. \label{fig:paymentupperbound}]{
			\begin{pspicture}(-1,-1)(6,4)
			\psset{yunit=3cm,xunit=1.6cm}
			\psaxes[Dy=1, Dx=1, labels=none, ticks=y]{->}(0, 0)(0,0)(4, 1.2)
			\pspolygon[fillstyle=solid, fillcolor=lightgray, linewidth=0pt](0,0)(0,0.93)(2,0.93)(2,0)				
			
			\pscustom[linewidth=0pt,fillstyle=solid,fillcolor=white]
			{ 	
				\psline[linewidth=1pt, linearc=0]{-}(0.8,0)(0.8, 0.5)
				\psbezier[liftpen=1,linewidth=1pt]{-}(0.8,0.5)(1.6,0.9)(1.6,1)(3.8,1)
				\psbezier[liftpen=1, linewidth=1pt, linestyle=dashed]{-}(3.8,1)(2.8,1)(2.6,0.9)(1.8,0.5) 
				\psline[liftpen=1,linewidth=1pt, linestyle=dashed]{-}(1.8, 0.5)(.8,0)
			}
			\rput(2.5, 1.1){$\qqbar(x)$}
			\rput(2.5,0.6){$\qcbar(x)$}
			\rput(2,-0.1){$\val$}
			\rput(3.5,-0.2){Value}
			\psline(0, -0.04)(0, 0.04)
			\psline(2, -0.04)(2, 0.04)		
		\end{pspicture}
			}	
			\hspace*{\fill}		
\caption{ \label{fig:opt-geometry}}

		\end{figure}


Recall that $\pricebar(\val)$ is the expected payment corresponding
with two-priced allocation rule~$\qqbar(\val)$.  We now give an upper
bound on~$\pricebar(\val)$.  

\begin{lemma}
\label{lem:pricebar}
The payment $\pricebar(\val)$ for $\val$ and two-priced rule
$\qqbar(\val)$ satisfies 
\begin{align}
\label{eq:pricebound}
\pricebar(\val) &\leq \val \qqbar(\val) -
\int_0^{\val} \qqbar(z) \: \dd z + \int_0^{\val} \qcbar(z) \: \dd z.
\end{align}
\end{lemma}
	
\begin{proof}
View a two-priced mechanism $\qqbar = \qvbar + \qcbar$ as charging
$\val$ with probability $\qqbar(\val)$ and giving a rebate of $\capa$
with probability $\qcbar(\val)$.  We bound this rebate as follows
(which proves the lemma):
\begin{align*}
\capa \cdot \qcbar(\val) &\geq \capa \cdot \qcbar(0) + \int_0^\val \qvbar(z) \: \dd z \\
                   &\geq \int_0^{\val} \qqbar(z) \: \dd z - \int_0^{\val} \qcbar(z) \: \dd z.
\end{align*}
The first inequality is from \autoref{thmpart:qcbar-bic} of \autoref{thm:mechbar}.  The second inequality
is from the definition of $\qcbar(0) = 0$ in \eqref{eq:qcbar} and
$\qvbar(\val) = \qqbar(\val) - \qcbar(\val)$. See \autoref{fig:opt-geometry} for an illustration.
\end{proof}


 
\subsection{Three-part Payment Decomposition}
\label{sec:revbar}

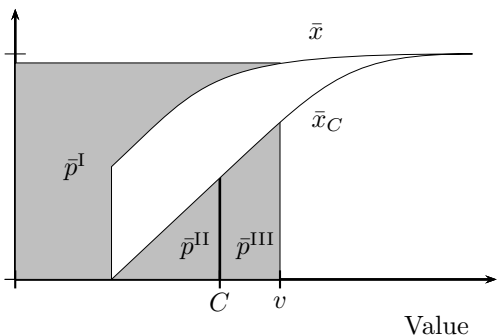
\begin{figure}[!hb]
			\small
			\centering
		\begin{pspicture}(-0.5,-1)(6.5,5)
			\psset{yunit=3cm,xunit=1.6cm}
			\psaxes[Dy=1, Dx=1, labels=none, ticks=y]{->}(0, 0)(0,0)(4, 1.2)
			\pspolygon[fillstyle=solid, fillcolor=lightgray, linewidth=0pt](0,0)(0,0.96)(2.2,0.96)(2.2,0)				
			\psline[linewidth=1pt, linearc=0](1.7,0)(1.7,0.6)

			\rput(0.5, 0.5){$\pI$}
			\rput(1.5, 0.15){$\pII$}
			\rput(2.0, 0.15){$\pIII$}
							
			\pscustom[linewidth=0pt,fillstyle=solid,fillcolor=white]
			{ 	
				\psline[linewidth=1pt, linearc=0]{-}(0.8,0)(0.8, 0.5)
				\psbezier[liftpen=1,linewidth=1pt]{-}(0.8,0.5)(1.6,0.9)(1.6,1)(3.8,1)
				\psbezier[liftpen=1, linewidth=1pt, linestyle=dashed]{-}(3.8,1)(2.8,1)(2.6,0.9)(1.8,0.5) 
				\psline[liftpen=1,linewidth=1pt, linestyle=dashed]{-}(1.8, 0.5)(.8,0)
			}
			\rput(2.5, 1.1){$\qqbar$}
			\rput(2.6,0.7){$\qcbar$}
			\rput(2.2,-0.1){$\val$}
			\rput(1.7,-0.1){$\capa$}
			\rput(3.5,-0.2){Value}
			\psline(0, -0.04)(0, 0.04)
			\psline(2.2, -0.04)(2.2, 0.04)
			\psline(1.7, -0.04)(1.7, 0.04)
		\end{pspicture}		
	\caption{Breakdown of the expected payment upper bound in a two-priced auction.\label{fig:payment-breakdown}}
\end{figure}

Below, we bound $\pricebar(\cdot)$ (and hence $\price(\cdot)$) in
terms of the expected payment of three natural mechanisms.  As seen
geometrically in \autoref{fig:payment-breakdown}, the bound given in
\autoref{lem:pricebar} can be broken into two parts: the area above
$\qqbar(\cdot)$, and the area below $\qcbar(\cdot)$.  We refer to the
former as $\pI(\cdot)$; we further split the latter quantity into two
parts: $\pII(\cdot)$, the area corresponding to $\val \in [0, \capa]$, and
$\pIII(\cdot)$, that corresponding to $\val \in [\capa, \val]$.  We
define these quantities formally below:

\begin{align}
\label{eq:pI}
\pI(\val) &= \qqbar(\val) \val - \int_0^{\val} \qqbar(z) \: \dd z,\\
\label{eq:pII}
\pII(\val) &= \int_0^{\min\{\val, \capa\}} \qcbar(z) \: \dd z \\
\label{eq:pIII}
\pIII(\val) &= \begin{cases}
0, & \val \leq \capa; \\
\int_{\capa}^{\val} \qcbar(z) \: \dd z, & \val > \capa.
\end{cases}
\end{align}

\begin{proof}[\stoccom{Proof} of \autoref{thm:pricebound}\stoccom{.}]
We now bound the revenue from each of the three parts of the payment
decomposition.  These bounds, combined with \autoref{cor:qcbar} of
\autoref{thm:mechbar} and \autoref{lem:pricebar}, immediately give
\autoref{thm:pricebound}.
\begin{description}
\item[Part~1.]
$\Ex[\val]{\pI(\val)} = \Ex[\val]{\virt(\val) \cdot \qqbar(\val)} \leq \Ex[\val]{\virtplus(\val) \cdot \qq(\val)}$.

Formulaically, $\pI(\cdot)$ corresponds to the risk-neutral payment
identity for $\qqbar(\cdot)$ as specified by 
\autoref{thmpart:payment} of \autoref{thm:myerson}; by 
\autoref{thmpart:virt} of \autoref{thm:myerson}, in expectation over
$\val$, this payment is equal to the expected virtual surplus
$\Ex[\val]{\virt(\val) \cdot \qqbar(\val)}$.\footnote{Note:  This
  equality does not require monotonicity of the allocation rule
  $\qqbar(\cdot)$; as long as \autoref{thmpart:payment} of
  \autoref{thm:myerson} formulaically holds, \autoref{thmpart:virt}
  follows from integration by parts.}  The inequality follows as terms
$\virt(\val)$ and $\qqbar(\val)$ in this expectation are point-wise
upper bounded by $\virtplus(\val) = \max(\virt(\val),0)$ and
$\qq(\val)$, respectively, the latter by \autoref{cor:qcbar} of \autoref{thm:mechbar}.

\item[Part~2.]
$\Ex[\val]{\pII(\val)} \leq \Ex[\val]{\virt(\val) \cdot \qcbar(\val)} \leq
\Ex[\val]{\virtplus(\val) \cdot \qc(\val)}$.

By definition of $\pII(\cdot)$ in \eqref{eq:pII}, if the statement of
the lemma holds for $\val = \capa$ it holds for $\val > \capa$; so we
argue it only for $\val \in [0,\capa]$.  Formulaically, with respect
to a risk-neutral agent with allocation rule $\qcbar(\cdot)$, the
risk-neutral payment is $\val \cdot \qcbar(\val) - \int_0^\val
\qcbar(z) \: \dd z$, the surplus is $\val \cdot \qcbar(\val)$, and the
risk-neutral agent's utility (the difference between the surplus and
payment) is $\int_0^\val \qcbar(z) \: \dd z = \pII(\val)$.  Convexity
of $\qcbar(\cdot)$, from \autoref{prop:qcbar-conv} of \autoref{thm:mechbar}, implies that the
risk-neutral payment is at least half the surplus, and so is at least
the risk-neutral utility. The lemma follows, then, by the same
argument as in the previous part.

\item[Part~3.]
$\Ex[\val]{\pIII(\val)} \leq \Ex[\val]{\pricevc \cdot \qcbar(\val)}
   = \Ex[\val]{\pricevc \cdot \qc(\val)}$.

The statement is trivial for $\val \leq \capa$ so assume $\val \geq
\capa$.  By definition $\qcbar(\val) = \qc(\val)$ for $\val > \capa$.
By \eqref{eq:bic-underbidding}, $\qc(\cdot)$ is monotone
non-decreasing.  Hence, for $\val > \capa$, $\pIII(\val) =
\int_{\capa}^{\val} \qc(z) \: \dd z\leq\int_{\capa}^{\val} \qc(\val)
\: \dd z = (\val-\capa) \cdot \qc(\val)$.
Plugging in $\pricevc = \max(\val - \capa,0)$ and taking expectation over~$\val$,
we obtain the bound.

\end{description}
\end{proof}

\section{Approximation Mechanisms and a Payment Identity}
\label{sec:payid}

In this section we first give a payment identity for Bayes-Nash
equilibria in mechanisms that charge agents a deterministic amount
upon winning (and zero upon losing).  Such one-priced payment schemes
are not optimal for capacitated agents; however, we will show that
they are approximately optimal.  When agents are symmetric (with
identical distribution and capacity) we use this payment identity to
prove that the first-price auction is approximately optimal.  When
agents are asymmetric we give a simple direct-revelation one-priced
mechanism that is BIC and approximately optimal.

\subsection{A One-price Payment Identity}

For risk-neutral agents, the Bayes-Nash equilibrium conditions entail
a payment identity: given an interim allocation rule, the payment rule
is fixed (\autoref{thm:myerson}, \autoref{thmpart:payment}).  For
risk-averse agents there is no such payment identity: there are
mechanisms with the identical BNE allocation rules but distinct BNE
payment rules.  We restrict attention to auctions wherein an agent's
payment is a deterministic function of her value (if she wins) and
zero if she loses.
We call these {\em one-priced} mechanisms; for these mechanisms there
is a (partial) payment identity.

Payment identities are an interim phenomenon.  We consider a single
agent and the induced allocation rule she faces from a Bayesian
incentive compatible auction (or, by the revelation principle, any BNE of any mechanism).
This allocation rule internalizes randomization in the environment and
the auction, and specifies the agents' probability of winning,
$\alloc(\val)$, as a function of her value. Given allocation
rule~$\alloc(\val)$, the risk-neutral expected payment is
$\pricern(\val) = \val \cdot \alloc(\val) - \int_0^{\val} \alloc(z) \:
\dd z$ (\autoref{thm:myerson}, \autoref{thmpart:payment}).
Given an allocation rule $\alloc(\val)$, a one-priced mechanism with
payment rule $\price(\val)$ would charge the agent
$\price(\val)/\alloc(\val)$ upon winning and zero otherwise (for an
expected payment of $\price(\val)$).  Define $\pvc(\val) = (\val -
\capa) \cdot \alloc(\val)$ which, intuitively, gives a lower bound on a
capacitated agent's willingness to trade-off decreased probability of
winning for a cheaper price.


\begin{theorem}
\label{thm:payid}
An allocation rule $\alloc$ and payment rule $\price$ are the BNE of a
one-priced mechanism if and only if (a) $\alloc$ is monotone
non-decreasing and (b) if $\price(\val) \geq \pvc(\val)$ for all $\val$
then $\price = \pricecap$ is defined as
\begin{align}
\pricecap(0) & = 0, 					\label{eq:payid-zero}\\
\pricecap(\val) & = \max 
	\left( \pvc(\val),\ 
		\sup_{\lowval < \val} \left\{ \pricecap(\lowval) + 
		(\pricern(\val) - \pricern(\lowval))\right\} 
	\right).								\label{eq:payid-max}
\end{align}
Moreover, if $\alloc$ is strictly increasing then $\price(\val) \geq
\pvc(\val)$ for all $\val$ and $\price = \pricecap$ is the unique
equilibrium payment rule.
\end{theorem}

The payment rule should be thought of in terms of two ``regimes'':
when $\pricecap = \pvc$, and when $\pricecap > \pvc$, corresponding 
to the first and second terms in the $\max$ argument of 
\eqref{eq:payid-max} respectively. In the latter regime, 
\eqref{eq:payid-max} necessitates that 
$\dpricecap(\val) = \dpricern(\val)$; for nearby such points $\val$
and $\val+\epsilon$, the $\lowval$ involved in the supremum will 
be the same, and thus $\pricecap(\val+\epsilon) - \pricecap(\val) = \pricern(\val+\epsilon) - \pricern(\val)$. 

The proof is relegated to \autoref{sec:payid-app}. The main intuition for this characterization is that risk-neutral
payments are ``memoryless'' in the following sense.  Suppose we fix
$\pricern(\val)$ for a~$\val$ and ignore the incentive of an agent with value
$\highval > \val$ to prefer reporting $\lowval < \val$, then the risk-neutral
payment for all $\highval > \val$ is $\pricern(\highval) = \price(\val) +
\int_\val^{\highval} (\alloc(\highval) - \alloc(z)) \, \dd z$.  This memorylessness is
simply the manifestation of the fact that the risk-neutral payment
identity imposes local constraints on the derivatives of the payment, i.e., $\dpricern(\val) = \val \cdot
\dalloc(\val)$.



There is a simple algorithm for constructing the risk-averse payment
rule $\pricecap$ from the risk-neutral payment rule $\pricern$ (for
the same allocation rule $\alloc$).
\begin{enumerate}
\item[0.] For $\val < \capa$,
$\pricecap(\val) = \pricern(\val)$. 
\item \label{step:p^C=>p^VC} The $\pricecap(\val) =
  \pricern(\val)$ identity continues until the value $\val'$ where
  $\pricecap(\val') = \pvc(\val')$, and 
  $\pricecap(\val)$ switches to follow $\pvc(\val)$.
\item When $\val$
increases to the value $\val''$ where $\dpricern(\val'') =
\dpvc(\val'')$ then $\pricecap(\val)$ switches to follow
$\pricern(\val)$ shifted up by the difference $\pvc(\val'') - \pricern(\val'')$ (i.e., its derivative $\dpricecap(\val)$ follows
$\dpricern(\val)$).  
\item Repeat this process from Step~\ref{step:p^C=>p^VC}.
\end{enumerate}

\begin{lemma}
\label{remark:payid}
The one-priced BIC allocation rule $\alloc$ and payment rule
$\pricecap$ satisfy the following
\begin{enumerate}
\item 
\label{thmpart:lbpc}
For all $\val$, $\pricecap(\val) \geq \max(\pricern(\val),\pvc(\val))$.
\item 
\label{thmpart:px-mon}
Both $\pricecap(\val)$ and $\pricecap(\val) / \alloc(\val)$ are monotone non-decreasing.
\end{enumerate}
\end{lemma}

The proof of \autoref{thmpart:px-mon} is contained in the proof of \autoref{lem:mono+payid=>bne} in \autoref{sec:payid-app}, and
\autoref{thmpart:lbpc} follows directly from equations \eqref{eq:payid-zero} and \eqref{eq:payid-max}.



\subsection{Approximate Optimality of First-price Auction}
\label{sec:firstprice}
We show herein that for agents with a common capacity and values drawn
i.i.d.\@ from a continuous, regular distribution $\dist$ with strictly
positive density the first-price auction is approximately optimal.

It is easy to solve for a symmetric equilibrium in the first-price
auction with identical agents.  First, guess that in BNE the agent
with the highest value wins.  When the agents are i.i.d.\@ draws from
distribution $\dist$, the implied allocation rule is $\alloc(\val) =
\dist^{n-1}(\val)$.  \autoref{thm:payid} then gives the necessary
equilibrium payment rule $\pricecap(\val)$ from which the bid function
$\bcap(\val) = \pricecap(\val) / \alloc(\val)$ can be calculated.  We
verify that the initial guess is correct as \autoref{remark:payid}
implies that the bid function is symmetric and monotone.  There is no
other symmetric equilibrium.\footnote{Any other symmetric equilibrum
  must have an allocation rule that is increasing but not always
  strictly so.  For this to occur the bid function must not be
  strictly increasing implying a point mass in the distribution of
  bids.  Of course, a point mass in a symmetric equilibrium bid
  function implies that a tie is not a measure zero event.  Any agent
  has a best response to such an equilibrium of bidding just higher
  than this pointmass so at essentially the same payment, she always
  ``wins'' the tie.}

\begin{proposition}
\label{prop:fpa-unique}
The first-price auction for identical (capacity and value
distribution) agents has a unique symmetric BNE wherein the highest
valued agent wins.
\end{proposition}

The expected revenue at this equilibrium is $n
\expect[\val]{\pricecap(\val)}$.  \autoref{remark:payid} implies that
$\pricecap$ is at least $\pricern$ and $\pvc$.

\begin{corollary} 
\label{cor:fpa-rev}
The expected revenue of the first-price auction for identical
(capacity and value distribution) agents is at least that of the
capacitated second-price auction and at least that of the second-price
auction.
\end{corollary}

Our main theorem then follows by combining \autoref{cor:fpa-rev} with
the revenue bound in \autoref{thm:pricebound} and \autoref{thm:bk} by
\citet{BK96}.

\begin{theorem} 
\label{thm:fpa-approx}
For $n\geq 2$ agents with common capacity and values drawn i.i.d.\@ from a
regular distribution, the revenue in the first price auction (FPA) in
the symmetric Bayes-Nash equilibrium is a 5-approximation to the
optimal revenue.
\end{theorem}

\begin{proof}
An immediate consequence of \autoref{thm:bk} is that for~$n\geq 2$
risk-neutral, regular, i.i.d.\@ bidders, the second-price auction
extracts a revenue that is at least half the optimal revenue; hence,
by \autoref{cor:revbound}, the optimal revenue for capacitated bidders
by any BIC mechanism is at most four times the second-price revenue
plus the capacitated second-price revenue.  Since the first-price
auction revenue in BNE for capacitated agents is at least the
capacitated second-price revenue and the second-price revenue, the
first-price revenue is a 5-approximation to the optimal
revenue.\footnote{In fact, the \citet{BK96} result shows that the
  second-price auction is asymptotically optimal so for large $n$ this
  bound can be asymptotically improved to three.}
\end{proof}

\subsection{Approximate Optimality of One-Price Auctions}
\label{sec:one-v-two}

We now consider the case of asymmetric value distributions and
capacities.  In such settings the highest-bid-wins first-price auction
does not have a symmetric equilbria and arguing revenue bounds for it
is difficult.  Nonetheless, we can give asymmetric one-priced
revelation mechanisms that are BIC and approximately optimal.
With respect to \autoref{ex:v<C} in \autoref{sec:pricebound} which
shows that the option to charge two possible prices from a given type
may be necessary for optimal revenue extraction, this result shows
that charging two prices over charging one price does not confer a
significant advantage.

\begin{theorem}
\label{thm:dc-3apx}
For $n$ (non-identical) agents, their capacities
$\capa_1,\ldots,\capa_n$, and regular value distributions
$\dist_1,\ldots,\dist_n$, there is a one-priced BIC mechanism whose revenue is at least one third of the optimal
(two-priced) revenue.
\end{theorem}

\begin{proof}
Recall from \autoref{thm:pricebound} that either the risk-neutral optimal revenue or $\Ex[\val_1, \ldots, \val_n]{\max
\{(\val_i - \capa_i)_+\}}$ is at least one third of the optimal revenue.  We apply \autoref{thm:payid} to two monotone
allocation rules:
\begin{enumerate}
\item the interim allocation rule of the risk-neutral optimal auction, and


\item the interim allocation rule specified by: serve agent~$i$ that
  maximizes $\val_i - \capa_i$, if positive; otherwise, serve nobody.
\end{enumerate}

As both allocations are monotone, we apply \autoref{thm:payid} to obtain two single-priced BIC mechanisms. 
By \autoref{remark:payid}, the expected revenue of the first mechanism is at least the risk neutral optimal revenue, and the expected
revenue of the second mechanism is at least $\Ex[\val_1, \ldots, \val_n]{\max \{(\val_i - \capa_i)_+\}}$.  The theorem
immediately follows for the auction with the higher expected revenue.
\end{proof}

Although \autoref{thm:dc-3apx} is stated as an existential result, the
two one-priced mechanisms in the proof can be described analytically
using the algorithm following \autoref{thm:payid} for calculating the capacitated BIC payment rule.  The interim allocation rules are straightforward
(the first: $\alloc_i(\val_i) = \prod_{j \neq i}
\dist_j(\virt_j^{-1}(\virt_i(\vali)))$, and the second:
$\alloc_i(\val_i) = \prod_{j \neq i} \dist_j(\val_i - \capa_i +
\capa_j)$), and from these we can solve for $\price^{\capa_i}_i(\val)$.
\section{Conclusions}
\label{s:conc}

%
%
For the purpose of keeping the exposition simple, we have applied our
analysis only to single-item auctions.  Our techniques, however, as
they focus on analyzing and bounding revenue of a single agent for a
given allocation rule, generalize easily to structurally rich
environments.  Notice that the main theorems of
Sections~\ref{sec:optimal}, \ref{sec:pricebound}, and the first part
of \autoref{sec:payid} do not rely on any assumptions on the
feasibility constraint except for downward closure, i.e., that it is
feasible to change an allocation by withholding service to an agent who was
formerly being served.

%
%
For example, our prior-independent 5-approximation result generalizes
to symmetric feasibility constraints such as position auctions.  A
position auction environment is given by a decreasing sequence of
weights $\alpha_1,\ldots,\alpha_n$ and the first-price position
auction assigns the agents to these positions greedily by bid.  With
probability $\alpha_i$ the agent in position $i$ receives an item and
is charged her bid; otherwise she is not charged.  (These position
auctions have been used to model pay-per-click auctions for selling
advertisements on search engines where $\alpha_i$ is the probability
that an advertiser gets clicked on when her ad is shown in the $i$th
position on the search results page.)  For agents with identical
capacities and value distributions, the first-price position auction
where the bottom half of the agents are always rejected is a
5-approximation to the revenue-optimal position auction (that may potentially
match all the agents to slots).

%
%
Our one- versus two-price result generalizes to asymmetric
capacities, asymmetric distributions, and asymmetric downward-closed
feasibility constraints.  A downward-closed feasibility constraint is
given by a set system which specifies which subset of agents can be
simultaneously served.  Downward-closure requires that any subset of a
feasible set is feasible.  A simple one-priced mechanism is a
3-approximation to the optimal mechanism in such an environment.  The
mechanism is whichever has higher revenue of the standard (risk
neutral) revenue-optimal mechanism (which serves the subset of agents
with the highest virtual surplus, i.e., sum of virtual values) and the
one-priced revelation mechanism that serves the set of agents $S$ that
maximizes $\sum_{i \in S} (\vali - \capa_i)^+$ subject to feasibility.

%
%
A main direction for future work is to relax some of the assumptions
of our model.  Our approach to optimizing over mechanisms for
risk-averse agents relies on (a) the simple model of risk aversion
given by capacitated utilities and (b) that losers neither make (i.e.,
ex post individual rationality) nor receive payments (i.e., no
bribes).  These restrictions are fundamental for obtaining linear
incentive compatibility constraints.  Of great interest in future
study is relaxation of these assumptions.

%
%
There is a relatively well-behaved class of risk attitudes known as
{\em constant absolute risk aversion} where the utility function is
parameterized by risk parameter $R$ as $\util_R(w) = \frac{1}{R} (1 -
e^{-Rw})$. These model the setting in which a bidder's risk aversion
is independent of wealth, and hence bidders view a lottery over
payments for an item the same no matter their valuations. \citet{M84}
exploits this and derives the optimal auction for such risk
attitudes. A first step in extending our results to more interesting
risk attitudes would be to consider such risk preferences.

%
%
Our analytical (and computational) solution to the optimal auction
problem for agents with capacitated utilities requires an {\em ex post
  individual rationality} constraint on the mechanism that is standard
in algorithmic mechanism design.  This constraint requires that an
agent who loses the auction cannot be charged.  While such a
constraint is natural in many settings, it is with loss and, in fact,
ill motivated for settings with risk-averse agents.  One of the most
standard mechanisms for agents with risk-averse preferences is the
``insurance mechanism'' where an agent who may face some large
liability with small probability will prefer to pay a constant
insurance fee so that the insurance agency will cover the large
liability in the event that it is incurred.  This mechanism is not ex
post individually rational.  Does the first-price
auction (which is ex post individual rational) approximate the
optimal interim individually rational mechanism?

\bibliographystyle{apalike}

\appendix
\section{Proofs from Section\IFSTOCELSE{ 3}{~\ref{sec:optimal}}}
\label{sec:opt-app}

\Xcomment{
\autoref{thm:optimal-two-price} (Restatement): For any auction for
capacitated agents there is a two-priced auction with no lower
revenue.

\begin{proof}[\NOTSTOC{Proof }of \autoref{thm:optimal-two-price}\NOTSTOC{.}]
As explained in the text, the proof consists of two steps.  We first
formalize the first step, then give the second in full.
\paragraph{Step~1\NOTSTOC{.}} \emph{There is a BIC optimal mechanism in which the wealth of any outcome is no more
than~$\capa$.}

Given any BIC mechanism~$\mech_0$, we construct the following
mechanism~$\mech_1$: $\mech_1$ does everything exactly the same way
as~$\mech_0$, except that whenever $\mech_0$ makes a truthful bidder's
wealth more than~$\capa$, $\mech_1$ makes it equal to~$\capa$.  To be
more precise, whenever $\mech_0$ sells the item and charges a bidder
reporting~$\val$ less than~$\val - \capa$, $\mech_1$ charges her $\val
- \capa$ instead.  $\mech_1$ has obviously at least as much revenue
as~$\mech_0$, since the payment in any outcome weakly increases.
$\mech_1$ is also BIC as long as~$\mech_0$ is: the utility of truthful
reporting is kept the same, but the utility for any type~$\val$ to
misreport another type~$\val'$ weakly decreases compared with
in~$\mech_0$.  

\paragraph{Step~2\NOTSTOC{.}} \emph{There is a BIC optimal mechanism in which the utility of any bidder is either
$\capa$ or~$0$ in any outcome.}

We take a BIC mechanism~$\mech_1$ in which a bidder in every outcome
has a nonnegative wealth no larger than~$\capa$.  We convert $\mech_1$
to a BIC mechanism~$\mech_2$ which has no less revenue but a bidder in
any outcome has wealth either $0$ or~$\capa$.  It is important that
in~$\mech_1$, only the linear part of the utility function is binding for
any truthful bidder.  Let $\ralloc$ and $\rprice_{\mech_1}$ be the
allocation and payment rule in~$\mech_1$.  We first define a parameter
in~$\mech_2$.  Fixing a type~$\val$, conditioning on $\mech_1$ selling
the item, the bidder's expected utility
is $\Ex[\rand]{\ucapa(\val - \rprice_{\mech_1}(\val)) \mid \ralloc(\val) = 1} \leq \capa$.  Define 
\begin{align*}
\beta(\val) = \frac{\Ex{\ucapa(\val - \rprice_{\mech_1}(\val)) \mid \ralloc(\val) = 1}}{\capa} \leq 1.
\end{align*}
By the property of~$\mech_1$, $\val - \rprice_{\mech_1}(\val) \leq \capa$ for any $\rand$, and
$\ucapa(\val - \rprice_{\mech_1}(\val)) = \val - \rprice_{\mech_1}(\val)$.  We therefore have
\begin{equation}
\label{eq:beta}
\Ex[\rand]{\rprice_{\mech_1}(\val) \given \ralloc(\val) = 1} = \val - \beta(\val) \capa.
\end{equation}

Now we construct mechanism~$\mech_2$.  On any reported value~$\val$,
$\mech_2$ sells the item if and only if~$\mech_1$ does; in addition,
conditioning on selling, $\mech_2$ charges the bidder $\val - \capa$
with probability~$\beta(\val)$, and with the remaining probability
charges her~$\val$.

The allocation rule of $\mech_1$ and~$\mech_2$ are the same, and we
denote both as $\ralloc$.  We use $\rprice_{\mech_2}$ to denote the
payment rules of~$\mech_2$.  Note that $\mech_2$ is so constructed
that any truthful bidder has the same expected utility in~$\mech_2$ as
in~$\mech_1$:
\begin{align*}
\Ex[\rand]{\ucapa(\val \ralloc(\val) - \rprice_{\mech_1}(\val))} = \Ex[\rand]{\ucapa(\val \ralloc(\val) -
\rprice_{\mech_2}(\val))}, \STOC{\\
&} \forall \val.
\end{align*}

We now consider:
\NOTSTOC{\begin{description}}

\IFSTOCELSE{\textbf{Revenue:}}{\item[Revenue:]} 
By definition, the expected revenues of the two mechanisms are the same.

\IFSTOCELSE{\textbf{Incentives:}}{\item[Incentives:]} 
It suffices to show that any agent's utility from misreporting (weakly) decreases.  I.e.,
$$\Ex[\rand]{\ucapa(\val \ralloc(\val') - \rprice_{\mech_2}(\val'))} \leq
\Ex[\rand]{\ucapa(\val \ralloc(\val') - \rprice_{\mech_1}(\val'))},$$ 
because then by the \eqref{eq:IC} condition
of~$\mech_1$ we have

Conditioning on not selling, the expected utility of type~$\val$ deviating to~$\val'$ is~$0$ in both.  We therefore only need to consider the expected utility conditioning on selling.  
In both mechanisms, the wealth of type~$\val$ misreporting~$\val'$ lies in $[\val - \val', \val - \val' + \capa]$
in any outcome.  Let $\mu$ be the linear function connecting
$(\val-\val', \ucapa(\val -\val'))$ and $(\val - \val' + \capa, \ucapa(\val - \val' + \capa))$, i.e.,
\begin{align*}
\mu(z) &= \ucapa(\val - \val')  \STOC{\\
& \qquad} + \frac{\ucapa(\val - \val' + \capa) - \ucapa(\val - \val')}{\capa} \cdot (z - \val + \val').
\end{align*}
We have
\begin{align*}
&\Ex[\rand]{\ucapa(\val - \rprice_{\mech_2}(\val')) \given \ralloc(\val') = 1} \\
&\qquad=  \beta(\val') \ucapa(\val - \val' + \capa) + (1 - \beta(\val')) \ucapa(\val - \val') \\
&\qquad=  \beta(\val') \mu(\val - \val' + \capa) + (1 - \beta(\val')) \ucapa(\val - \val') \\
&\qquad=  \mu(\val - \val' + \beta(\val') \capa) \\
&\qquad=  \Ex[\rand]{\mu(\val - \rprice_{\mech_1}(\val')) \given \ralloc(\val') = 1} \\
&\qquad\leq  \Ex[\rand]{\ucapa(\val - \rprice_{\mech_1}(\val')) \given \ralloc(\val') = 1}.
\end{align*}
The first equality is by construction of~$\mech_2$; the second from the definition of $\ucapa$ and~$\mu$; the third by
the linearity of~$\mu$; the fourth by~\eqref{eq:beta}.  The last inequality is because, on the interval $[\val - \val', \val
- \val' + \capa]$, $\mu$ is pointwise weakly dominated by~$\ucapa$ by the concavity of~$\ucapa$.

This shows that in~$\mech_2$, any bidder of any type has weakly less incentive to misreport his valuation than in~$\mech_1$.
\NOTSTOC{\end{description}}

This finishes the proof of \autoref{thm:optimal-two-price}.
\end{proof}

}

\noindent{\bf \autoref{lem:two-price-BIC} (Restatement).} {\em A mechanism
with two-price allocation rule $\qq = \qv - \qc$ is BIC if and only if
for all $\val$ and $\highval$ such that $\val < \highval \leq \val +
\capa$,
\begin{equation}
\frac{\qv(\val)}{\capa} \leq \frac{ \qc(\highval) - \qc(\val)}{\highval - \val} \leq \frac{\qq(\highval)}{\capa}.\tag{\ref{eq:near-bic}}
\end{equation}}

\begin{proof}[\NOTSTOC{Proof }of \autoref{lem:two-price-BIC}\NOTSTOC{.}]
%

Consider an agent and fix two possible values of the agent $\val \leq
\highval \leq \val + \capa$.  The utility for truthtelling with
value~$\val$ is $\capa \cdot \qc(\val)$ in a two-price auction.  The
utility for misreporting $\highval$ from value~$\val$ is
$\qv(\highval) \cdot (\val - \highval) + \qc(\highval) \cdot (\capa +
\val - \highval)$: when the mechanism sells and charges $\highval$,
the agent's utility is~$\val - \highval$; when the mechanism sells and
charges $\highval - \capa$, her utility is $\ucapa(\capa + \val -
\highval) = \capa + \val - \highval$ (since $\val < \highval$).
Likewise, the utility for misreporting $\val$ from true value
$\highval$ is $\qv(\val) \cdot (\highval - \val) + \qc(\val) \cdot
\capa$.  Note that here when the mechanism charges $\val - \capa$, the
utility of the agent is~$\capa$ because the wealth $\capa - \val +
\highval$ is more than~$\capa$; when the mechanism charges~$\val$, her
utility is $\highval - \val$ because we assumed $\highval \leq \val +
\capa$.

An agent with valuation~$\val$ (or $\highval$) would not misreport
$\highval$ (or $\val$) if and only if
\begin{align}
\qc(\val) \cdot \capa \geq &\  \qv(\highval) \cdot (\val - \highval) \STOC{\nonumber\\
						& \quad} + \qc(\highval) \cdot (\capa + \val - \highval);  \label{eq:v-dev-v'}\\
\qc(\highval) \cdot \capa \geq &\ \qv(\val) \cdot (\highval - \val) + \qc(\val)  \cdot \capa. \label{eq:v'-dev-v}
\end{align}
Now the right side of \eqref{eq:near-bic} follows
from~\eqref{eq:v-dev-v'} and the left side follows
from~\eqref{eq:v'-dev-v}.

When $\highval > \val + \capa$, the agent with value~$\val$ certainly
has no incentive to misreport~$\highval$, since any outcome results in
non-positive utility. Alternatively, the agent with value $\highval$
will derive utility $\capa \cdot \qq(\val)$ from misreporting $\val$
and thus will misreport if and only if $\qq(\val) >
\qc(\highval)$. Substituting $\val+\capa$ for $\highval$ in
equation~\eqref{eq:near-bic} gives $\qq(\val)\leq \qc(\val + \capa)$,
and taking this for intermediate points between $\val + \capa$ and
$\highval$ gives monotonicity of $\qc(\val)$ over $[\val + \capa,
  \highval]$. Combining these gives $\qq(\val)\leq \qc(\val + \capa)
\leq \qc(\val)$ and hence $\highval$ will not misreport $\val$.
%
\end{proof}

\noindent{\bf \autoref{cor:int-bic} (Restatement):}
{\em The allocation rules $\qc$ and $\qv$ of a BIC two-priced mechanism satisfies that for all $\val < \highval$, 
\begin{align}
 \int_{\val}^{\highval} \frac{\qv(z)}{\capa} \: \dd z \leq \qc(\highval) - \qc(\val) \leq \int_{\val}^{\highval}
\frac{\qq(z)}{\capa} \: \dd z.\tag{\ref{eq:int-bic}}
\end{align}}

\begin{proof}[\NOTSTOC{Proof }of \autoref{cor:int-bic}\NOTSTOC{.}]
Without loss of generality, suppose $\highval \leq \val + \capa$ (the
statement then follows for higher $\highval$ by induction).  Define
function
\begin{align*}
\qcbar(z) = \qc(\val) + \int_{\val}^{z} \frac{\sup_{y' \in [\val, y]} \qv(y')}{\capa} \: \dd y, \quad \forall z \in
[\val, \highval],
\end{align*}
then $\qcbar(z) \geq \qc(\val) + \int_{\val}^z \frac{\qv(y)}{\capa} \: \dd y$ and hence

\begin{align*}
\int_{\val}^z \frac{\qv(y)}{\capa} \: \dd y \leq \qcbar(z) - \qc(\val).
\end{align*}
By the argument in the proof of \autoref{thm:mechbar},\autoref{lem:qcbar}, we have $\qcbar(z) \leq \qc(z)$, for all~$z$.  This gives the left
side of \eqref{eq:int-bic}.  The other side is proven similarly.
\end{proof}

\section{Proofs from \IFSTOCELSE{Section~4}{\autoref{sec:pricebound}}}
\label{sec:pricebound-app}

\noindent{\bf \autoref{def:mechbar} (Restatement).}
We define $\qqbar = \qcbar + \qvbar$ as follows:
\begin{enumerate}
\item $\qvbar(\val) = \qv(\val)$;
\item Let $\crate(\val)$ be $\tfrac{1}{\capa} \sup_{z \leq \val} \qv(z)$, and let
\begin{align}
\tag{\ref{eq:qcbar}}
\qcbar(\val) = \left\{ \begin{array}{ll}
\int_0^{\val} \crate(y) \: \dd y, & \val \in [0, \capa]; \\
\qc(\val), & \val > \capa.
\end{array}
\right.
\end{align}
\end{enumerate}

\noindent{\bf \autoref{thm:mechbar} (Restatement).} {\em \begin{enumerate}
\item 
On $\val \in [0,\capa]$, $\qcbar(\cdot)$ is a convex, monotone
increasing function.

\item 
On all $\val$, $\qcbar(\val) \leq \qc(\val)$.

\item 
The incentive constraint from the left-hand side of \eqref{eq:int-bic} holds for $\qcbar$:
$\int_\val^{\highval} \qvbar(z)\: \dd z \leq \qcbar(\highval) - \qcbar(\val)$ for all $\val < \highval$.

\item 
On all $\val$, $\qcbar(\val) \leq \qc(\val)$, $\qqbar(\val) \leq \qq(\val)$, and $\pricebar(\val) \geq \price(\val)$.
\end{enumerate}}

\begin{proof}[\stoccom{Proof} of \autoref{thm:mechbar}\stoccom{.}]
\ 
\begin{enumerate}
\item On $[0,\capa]$, $\qcbar(\val)$ is the integral of a monotone,
non-negative function.

\item The statement holds directly from the definition for $\val >
  \capa$; therefore, fix $\val \leq \capa$ in the argument below.

Since $\crate(\val)$ is an increasing function of~$\val$, it is
Riemann integrable (and not only Lebesgue integrable).

Fixing $\val$, we show that, given any $\eps \leq 0$, $\qcbar(\val) \leq \qc(\val) + \eps$.  Fix an integer $\npart >
\val / \eps$, and let $\Delta$ be $\val / \npart < \eps$.  Consider Riemann sum $\riemsum = \sum_{j = 1}^{\npart} \Delta \cdot \crate(\xi_j)$,
where each $\xi_j$ is an arbitrary point in $[(j-1)\Delta, j \Delta]$.\footnote{Obviously $\riemsum$ depends both on
$\Delta$ and the choice of $\xi_j$'s.  For cleanness of notation we omit this dependence and do not write
$\riemsum_{\Delta, \xi}$.}  We will also denote by $\riemsum(k) = \sum_{j = 1}^k \Delta \cdot
\crate(\xi_j)$, $k \leq \npart$, the partial sum of the first $k$ terms.  Since $\qcbar(\val) = \lim_{\Delta \to 0}
\riemsum$, it suffices to show that for all $\Delta < \eps$, $\riemsum \leq \qc(\val)
+ \eps$.  In order to show this, we define a piecewise linear function $\lowqc$.  On $[0, \Delta]$, $\lowqc$ is~$0$, and
then on interval $[j \Delta, (j + 1)\Delta]$, $\lowqc$ grows at a rate $\crate((j-1)\Delta)$.  Intuitively, $\lowqc$
``lags behind'' $\qc$ by an interval $\Delta$ and we will show it lower bounds $\qc$ and upper bounds $\riemsum + \eps$.
Note that since $\crate$ is an increasing function, $\lowqc$ is convex.



We first show $\lowqc(\val) \leq \qc(\val)$.  We will show by induction on~$j$ that $\lowqc(z) \leq \qc(z)$ for
all $z \in [0, j\Delta]$.  Since $\lowqc$ is $0$ on $[0, \Delta]$, the base case $j = 1$ is trivial.  Suppose we have
shown $\lowqc(z) \leq \qc(z)$ for all $z \in [0, (j-1)\Delta]$, let us consider the interval $[(j-1)\Delta, j\Delta]$.
Let $z^*$ be $\argmax_{z \leq (j-1)\Delta} \qv(z)$.\footnote{Here we assumed that $\sup_{z < (j-1)\Delta} \qv(z)$
can be attained by~$z^*$, which is certainly the case when $\qv$ is continuous.  It is straightforward to see though
that we do not need such an assumption.  It suffices to choose $z^*$ such that $\qv(z^*)$ is close enough to $\crate((j-1)\Delta)$.  The proof goes almost without change,
except with an even smaller choice of~$\Delta$.}  By the induction hypothesis, $\lowqc(z^*) \leq \qc(z^*)$.  Recall that
$z^* \leq z \leq \capa$.  By the BIC
condition \eqref{eq:bic-underbidding}, for all $z \geq z^*$, 
\begin{align*}
\qc(z) \geq \qc(z^*) + \frac{\qv(z^*)}{\capa} (z - z^*).
\end{align*}
On the other hand, 
by definition, $\crate$ is constant on $[z^*, z]$, and the derivative of $\lowqc$ is no larger than $\crate(z^*)$ on
$[z^*, z]$.  Hence for all $z \leq j\Delta$, 
\begin{align*}
\lowqc(z) & \leq \lowqc(z^*) + \frac{\qv(z^*)}{\capa}(z - z^*) \\
& \leq \qc(z^*) + \frac{\qv(z^*)}{\capa} (z - z^*)
\leq \qc(z).
\end{align*}
This completes the induction and shows $\lowqc(z) \leq \qc(\val)$ for all $z \in [0, \val]$.

Now we show $\riemsum\leq \lowqc(\val) + \eps$.  Note that since $\crate(z) \leq 1$ for all~$z$, $\riemsum
\leq \riemsum(\npart - 1) + \Delta < \riemsum(\npart - 1) + \eps$.  We will show by induction that $\riemsum(\npart -
1) \leq \lowqc(\val)$.  Our induction hypothesis is $\riemsum(j - 1) \leq \lowqc(j \Delta)$.  The base case
for $j = 1$ is obvious as $\riemsum(0) = \lowqc(\Delta) = 0$.
\begin{align*}
\riemsum(j) & = \riemsum(j-1) + \Delta \cdot \crate(\xi_j) \\
& \leq \lowqc(j) + \Delta \cdot \crate(j \Delta) \\
& = \lowqc(j+1).
\end{align*}
In the inequality we used the induction hypothesis and the monotonicity of~$\crate$.  The last equality is by definition
of~$\lowqc$.

This completes the proof of \autoref{lem:qcbar}.

\item 
For $\val\leq\highval \leq \capa$, by definition of $\qcbar$,
\begin{align*}
\qcbar(\highval) - \qcbar(\val) = \int_{\val}^{\highval} \crate(z) \: \dd z \geq \int_{\val}^{\highval}
\frac{\qv(z)}{\capa} \: \dd z.
\end{align*}
For $\capa \leq \val\leq \highval$, $\qcbar$ and $\qvbar$ are equal to $\qcbar$ and $\qvbar$ on $[\val, \highval]$, and the
inequality follows from \autoref{cor:int-bic}.
For $\val \leq C$ and $\highval \geq C$, we have
\begin{align*}
\qcbar(\highval) - \qcbar(\val) & = [\qcbar(\highval) - \qcbar(\capa)] + [\qcbar(\capa) - \qcbar(\val)]\\
& \geq \int_{\capa}^{\highval} \frac{\qv(z)}{\capa} \: \dd z + \int_{\val}^{\capa} \frac{\qv(z)}{\capa} \: \dd z \\
& = \int_{\val}^{\highval} \frac{\qv(z)}{\capa} \: \dd z.
\end{align*}


\item The first part, $\qcbar(\val) \leq \qc(\val)$, is from
  \autoref{lem:qcbar} of the lemma and the definition of $\qcbar(\val)
  = \qc(\val)$ on $\val > C$.  The second part, $\qqbar(\val) \leq
  \qq(\val)$, follows from the definition of $\qvbar(\val) =
  \qv(\val)$, the first part, and the definition of $\qq(\val) =
  \qv(\val) + \qc(\val)$.  The third part, $\pricebar(\val) \geq
  \price(\val)$, follows because lowering $\qc(\val)$ to
  $\qcbar(\val)$ on $\val \in [0,\capa]$ foregoes payment of
  $\val-\capa$ which is non-positive (for $\val \in [0,\capa]$). \qedhere
\end{enumerate}
\end{proof}

%
%
%
%

\section{Proofs from \IFSTOCELSE{Section 5}{\autoref{sec:payid}}}
\label{sec:payid-app}
\label{sec:first-app}

\noindent{\bf \autoref{thm:payid} (Restatement).}
{\em An allocation rule $\alloc$ and payment rule $\price$ are the BNE of a one-priced mechanism if and only if 
(a) $\alloc$ is 
monotone non-decreasing and (b) if $\price(\val) \geq \pvc(\val)$ for all $\val$ then $\price = \pricecap$ is defined as
\begin{align}
\tag{\ref{eq:payid-zero}}
\pricecap(0) & = 0,\\
\tag{\ref{eq:payid-max}}
\pricecap(\val) & = \max 
	\left( \pvc(\val),\ 
		\sup_{\lowval < \val} \left\{ \pricecap(\lowval) + 
		(\pricern(\val) - \pricern(\lowval))\right\} 
	\right).								
\end{align}
Moreover, if $\alloc$ is strictly increasing then $\price(\val) \geq
\pvc(\val)$ for all $\val$ and $\price = \pricecap$ is the unique
equilibrium payment rule.}

\vspace{.3cm}

The proof follows from a few basic conditions. First, with strictly
monotone allocation rule $\alloc$, the payment upon winning must be at
least $\val-\capa$; otherwise, a bidder would wish to overbid and see
a higher chance of winning, with no decrease in utility on
winning. Second, when the payment on winning is strictly greater than
$\val-\capa$, the bidder is effectively risk-neutral and the
risk-neutral payment identity must hold locally. Third, when an agent
is paying exactly $\val-\capa$ on winning, they are capacitated when
considering underbidding, but risk-neutral when considering
overbidding. As a result, at such a point, $\pricecap$ must be at
least as steep as $\pricern$, i.e., if $\dpricern(\val) >
\dpvc(\val)$, $\pricecap$ will increase above $\pvc$, at which point
it must follow the behavior of $\pricern$.


\autoref{thm:payid} follows from the following three lemmas which show
the necessity of monotonicity, the (partial) necessity of the payment
identity, and then the sufficiency of monotonicity and the payment
identity.

\begin{lemma} 
\label{lem:bne=>mono}
If $\alloc$ and $\price$ are the BNE of a one-priced mechanism, then
$\alloc$ is monotone non-decreasing.
\end{lemma}

\autoref{lem:bne=>mono} shows that monotonicity of the allocation rule
is necessary for BNE in a one-priced mechanism.  Compare this to
\autoref{ex:nonmono} where we exhibited a non-one-priced mechanisms
that was not monotone.  Because the utilities may be capacitated, the
standard risk-neutral monotonicity argument; which involves writing
the IC constraints for a high-valued agent reporting low and a
low-valued agent reporting high, adding, and canceling payments; does
not work.

\begin{lemma} 
\label{lem:bne=>payid}
If $\alloc$ and $\price$ are the BNE of a one-priced mechanism and
$\price(\val) \geq \pvc(\val)$ for all $\val$, then $\price =
\pricecap$ (as defined in \autoref{thm:payid}); moreover, if $\alloc$ is
strictly monotone then $\price(\val) \geq \pvc(\val)$ for all $\val$.
\end{lemma}

From \autoref{lem:bne=>payid} we see that one-priced mechanisms almost
have a payment identity.  It is obvious that a payment identity does
not generically hold as a capacitated agent with value $\val$ is
indifferent between payments less than $\val - \capa$; therefore, the
agent's incentives does not pin down the payment rule if the payment
rule ever results in a wealth for the agent of more than $\capa$.
Nonetheless, the lemma shows that this is the only thing that could
lead to a multiplicity of payment rules.  Additionally, the lemma
shows that if $\alloc$ is strictly monotone, then these sorts of
payment rules cannot arise.

\begin{lemma} 
\label{lem:mono+payid=>bne}
If allocation rule $\alloc$ is monotone non-decreasing and payment
rule $\price = \pricecap$ (as defined in \autoref{thm:payid}), then
they are the Bayes-Nash equilibrium of a one-priced mechanism.
\end{lemma}

The following claim and notational definition will be used throughout the proofs below.


\begin{claim}
\label{c:misreports}
Compared to the wealth of type $\val$ on truthtelling, when type
$\highval > \val$ misreports $\val$ she obtains strictly more wealth
(and is more capacity constrained) and when type $\lowval < \val$
misreports $\val$ she obtains strictly less wealth (and is less
capacity constrained) and if $\price(\val) \geq \pvc(\val)$ then
type $\lowval$ is strictly risk neutral on reporting $\val$.
\end{claim}

\begin{definition}
Denote the utility for type $\val$ misreporting $\val'$ for the same implicit
allocation and payment rules by $\ucapareport(\val, \val')$ and
$\urnreport(\val, \val')$ for risk-averse and risk-neutral agents,
respectively.
\end{definition}

\begin{proof}[\stoccom{Proof} of \autoref{lem:bne=>mono}\stoccom{.}]
We prove via contradiction. Assume that $\alloc$ is not monotone, and
hence there is a pair of values, $\lowval < \highval$, for which $\alloc(\lowval)> \alloc(\highval)$. We will
consider this in three cases: (1) when a type of $\lowval$ is
capacitated upon truthfully reporting and winning, and when a type of
$\lowval$ is strictly in the risk-neutral section of her utility upon
winning and a type of $\highval$ is either in the (2) capacitated or
(3) strictly risk-neutral section of her utility upon winning.

\begin{enumerate}
\item ($\lowval$ capacitated).  If $\lowval$ is capacitated upon
  winning, then $\highval$ will also be capacitated upon winning and
  misreporting $\lowval$ (\autoref{c:misreports}).  A capacitated agent is already receiving
  the highest utility possible upon winning.  Therefore, $\highval$
  strictly prefers misreporting $\lowval$ as such a report (strictly)
  increases probability of winning and (weakly) increases utility from
  winning.

\item ($\lowval$ risk-neutral, $\highval$ capacitated).  We split this case into two subcases depending on whether the agent with type $\lowval$ is capacitated with misreport $\highval$.
\begin{enumerate}
\item ($\lowval$ capacitated when misreporting $\highval$). As the
  truthtelling $\highval$ type is also capacitated (by assumption of
  this case), the utilities of these two scenarios are the same,
  i.e.,
\begin{align}
\label{eq:mono1}
\ucapareport(\lowval, \highval) &=
\ucapareport(\highval, \highval).\\
\intertext{Since type $\lowval$ truthfully reporting
$\lowval$ is strictly uncapacitated, if her value was increased she
would feel a change in utility (for the same report); therefore,
type $\highval$ reporting $\lowval$ has strictly more utility (\autoref{c:misreports}), i.e.,}
\label{eq:mono2}
\ucapareport(\highval, \lowval) &> \ucapareport(\lowval, \lowval).\\
\intertext{Combining \eqref{eq:mono1} and \eqref{eq:mono2} we arrive at the contradiction that type $\highval$ strictly prefers to report $\lowval$, i.e.,}
\notag
\ucapareport(\highval, \lowval) &> \ucapareport(\highval, \highval).
\end{align}

\item ($\lowval$ risk-neutral when misreporting $\highval$).  First,
  it cannot be that the bidder of type $\highval$ is capacitated for
  both reports $\highval$ and $\lowval$ as, otherwise, misreporting
  $\lowval$ gives the same utility upon winning but strictly higher
  probability of winning.  Therefore, both types are risk neutral when
  reporting $\lowval$. Type $\lowval$ is risk-neutral for both reports
  so she feels the discount in payment from reporting $\highval$
  instead of $\lowval$ linearly; type
  $\highval$ feels the discount less as she is capacitated at
  $\highval$.  On the other hand, $\highval$ has a higher value for
  service and therefore feels the higher service probability from
  reporting $\lowval$ over $\highval$ more than $\lowval$.
  Consequently, if $\lowval$ prefers reporting $\lowval$ to
  $\highval$, then so must $\highval$ (strictly).
\end{enumerate}


%
%

\item ($\lowval$ risk-neutral, $\highval$ risk-neutral).  First, note
  that the price upon winning must be higher when reporting $\lowval$
  than $\highval$, i.e., $\price(\lowval)/\alloc(\lowval) >
  \price(\highval)/\alloc(\highval)$; otherwise a bidder of type
  $\highval$ would always prefer to report $\lowval$ for the higher
  utility upon winning and higher chance of winning.  Thus, a bidder
  of type $\highval$ must be risk-neutral upon underreporting
  $\lowval$ and winning; furthermore, risk-neutrality of $\highval$
  for reporting $\highval$ implies the risk-neutrality of $\lowval$
  for reporting $\highval$ (\autoref{c:misreports}).  As both
  $\highval$ and $\lowval$ are risk-neutral for reporting either of
  $\lowval$ or $\highval$, the standard monotonicity argument for
  risk-neutral agents applies.

\end{enumerate}

Thus, for $\alloc$ to be in BNE it must be monotone non-decreasing.
\end{proof}

\begin{proof}[\stoccom{Proof} of \autoref{lem:bne=>payid}\stoccom{.}]
First we show that if $\alloc$ is strictly monotone then
$\price(\val) \geq \pvc(\val)$ for all $\val$.  If $\price(\val) <
\pvc(\val)$ then type $\val$ on truthtelling obtains a wealth
$w$ strictly larger than $\capa$.  Type $\lowval = \val -
\eps$, for $\eps \in (0,w-\capa)$, would also be capacitated
when reporting $\val$; therefore, by strict monotonicity of $\alloc$
such a overreport strictly increases her utility and BIC is violated.

The following two claims give the necessary condition.
\begin{align}
\label{eq:payid-low}
\pricecap(\val) & \geq \pricecap(\lowval) + (\pricern(\val) - \pricern(\lowval)), \quad \forall \lowval < \val  \\
\label{eq:payid-high}
\pricecap(\val) & \leq \sup_{\lowval < \val} \left\{ \pricecap(\lowval) + (\pricern(\val) - \pricern(\lowval)) \right\}, \quad \forall \val {\textrm \ s.t.\ }
\pricecap(\val) > \pvc(\val).
\end{align}

Equation \eqref{eq:payid-low} is easy to show.  Since $\pricecap(\val)
\geq \pvc(\val)$, the wealth of any type~$\lowval$ when winning is at
most~$\capa$, and strictly smaller than $\capa$ if overbidding.  In
other words, when overbidding, a bidder only uses the linear part of
her utility function and therefore can be seen as risk neutral.
Equation \eqref{eq:payid-low} then follows directly from the standard
argument for risk neutral agents.\footnote{For a risk neutral agent,
  the risk neutral payment maintains the least difference in payment
  to prevent all types from overbidding.}

Equation \eqref{eq:payid-high} would be easy to show if $\pricecap$ is
continuous: for all $\val$ where $\pricecap(\val) > \pvc(\val)$, there
is a neighborhood $(\val - \eps, \val]$ such that deviating on this
  interval only incurs the linear part of the utility function and the
  agent is effectively risk neutral.  We give the following general
  proof that deals with discontinuity and includes continuous cases as
  well.

To show \eqref{eq:payid-high}, it suffices to show that, for each $\val$ where $\pricecap(\val)
> \pvc(\val)$, for any $\eps > 0$, $\pricecap(\val) < \pricecap(\lowval) + (\pricern(\val) - \pricern(\lowval) + \eps$
for some $\lowval < \val$.  Consider any $\lowval > \val - \tfrac \eps 2$.  Since $\pricecap(\lowval) \geq \pvc(\lowval) =
(\lowval - \capa) \alloc(\lowval) > (\val - \tfrac \eps 2 - \capa) \alloc(\lowval)$, the utility for $\val$ to misreport
$\lowval$, i.e., $\ucapareport(\val, \lowval)$ is not much smaller than if the agent is risk neutral:
\begin{align*}
\urnreport(\val, \lowval) - \ucapareport(\val, \lowval) < \frac \eps 2 \alloc(\lowval).
\end{align*}
The following derivation, starting with the BIC condition, gives the desired bound:
\begin{align*}
0 \leq \ucapareport(\val, \val) - \ucapa(\val, \lowval) & < \ucapa(\val, \val) - \urnreport(\val, \lowval) + \frac \eps
2 \alloc(\lowval) \\
& = (\alloc(\val) \val - \pricecap(\val)) - (\alloc(\lowval) \val - \pricecap(\lowval)) + \frac \eps 2 \alloc(\lowval) \\
& = (\alloc(\val) - \alloc(\lowval))\val - (\pricecap(\val) - \pricecap(\lowval)) + \frac \eps 2 \alloc(\lowval) \\
& \leq \pricern(\val) - \pricern(\lowval) + (\val - \lowval) \alloc(\val) - (\pricecap(\val) - \pricecap(\lowval)) +
\frac \eps 2 \alloc(\lowval) \\
& \leq \pricern(\val) - \pricern(\lowval) - (\pricecap(\val) - \pricecap(\lowval)) + \eps.
\end{align*}
The first equality holds because $\pricecap(\val) > \pvc(\val)$; the second to last inequality uses the definition of risk
neutral payments (\autoref{thm:myerson}, \autoref{thmpart:payment}), and the last holds because $\alloc(\lowval)
< \alloc(\val) \leq 1$.
\end{proof}

\begin{proof}[\stoccom{Proof} of \autoref{lem:mono+payid=>bne}\stoccom{.}]
The proof proceeds in three steps.  First, we show that an agent with
value $\val$ does not want to misreport a higher value $\highval$.
Second, we show that the expected payment on winning, i.e.,
$\pricecap(\val)/\alloc(\val)$ is monotone in $\val$.  Finally, we
show that the agent with value $\val$ does not want to misreport a
lower value $\lowval$. Recall in the subsequent discussion that
$\pricern$ is the risk-neutral expected payment for allocation
rule~$\alloc$ (from \autoref{thm:myerson}, \autoref{thmpart:payment}).


\begin{enumerate} 
\item \label{step:misreporting-highval} (Type $\val$ misreporting $\highval$.)  This argument pieces
  together two simple observations.  First, \autoref{c:misreports}
  and the fact that $\pricecap \geq \pvc$ imply that $\val$ is
  risk-neutral upon reporting $\highval$.  
Second, by definition of $\pricecap$, the difference in a capacitated agent's
  payments given by $\pricecap(\highval) - \pricecap(\val)$ is at least that for a risk neutral agent given by
  $\pricern(\highval) - \pricern(\val)$.  The risk-neutral agent's
  utility is linear and she prefers reporting $\val$ to $\highval$.
  As the risk-averse agent's utility is also linear for payments in
  the given range and because the difference in payments is only
  increased, then the risk-averse agent must also prefer reporting
  $\val$ to $\highval$.

\item (Monotonicity of $\pricecap / \alloc$.)  The monotonicity of
  $\tfrac{\pricecap}{\alloc}$, which is \autoref{thmpart:px-mon} of
  \autoref{remark:payid}, will be used in the next case (and some
  applications of \autoref{thm:payid}).  We consider $\val$ and
  $\highval$ and argue that $\tfrac{\pricecap(\val)}{\alloc(\val)}
  \leq \tfrac{\pricecap(\highval)}{\alloc(\highval)}$.  First, suppose
  that the wealth upon winning of an agent with value $\val$ is
  $\capa$, i.e., $\pricecap(\val) = \pvc(\val)$.  If
  $\pricecap(\highval) = \pvc(\highval)$ as well, then by definition of $\pvc$
  (by $\tfrac{\pvc(\val)}{\alloc(\val)} = \val - \capa$) monotonicity
  of $\pricecap/\alloc$ holds for these points.  If $\pricecap$ is
  higher than $\pvc$ at $\highval$ then this only improves
  $\pricecap/\alloc$ at $\highval$.  Second, suppose that the wealth
  of an agent with value $\val$ is strictly larger than $\capa$,
  meaning this agent's utility increases with wealth.  The allocation
  rule $\alloc(\cdot)$ is weakly monotone (\autoref{lem:bne=>mono}), suppose for a contradiction
  that $\tfrac{\pricecap(\val)}{\alloc(\val)} >
  \tfrac{\pricecap(\highval)}{\alloc(\highval)}$ on $\val < \highval$.
  Then the agent with value $\val$ can pretend to have value
  $\highval$, obtain at least the same probability of winning, and
  obtain strictly lower payment.  This increase in wealth is strictly
  desired, and therefore, this agent strictly prefers misreport
  $\highval$.  Combined with \autoref{step:misreporting-highval},
  above, which argued that a low valued agent would not prefer to
  pretend to have a higher value, this is a contradiction.


\item (Type $\val$ misreporting $\lowval$.)  If $\pricecap(\val) =
  \pvc(\val)$, then paying less on winning does not translate into
  extra utility, and hence by the monotonicity of $\pricecap/\alloc$,
  the agent would never misreport.

We thus focus then on the case that $\pricecap(\val) > \pvc(\val)$. By
the monotonicity of $\pricecap / \alloc$, there is a point $\valzero <
\val$ such that for every value $\lowval$ between $\valzero$ and $\val$, if
an agent with value $\val$ reported $\lowval$, she would still be in
the risk-neutral section of her utility function. Specifically, this
entails that $\forall \lowval$ such that $\valzero < \lowval < \val$,
$\pricecap(\lowval)/\alloc(\lowval) \geq \val - \capa$. Consider such
a $\valzero$ and any such $\lowval$. For any such point,
$\pricecap(\lowval)/\alloc(\lowval) > \lowval - \capa$, and hence a
bidder with value $\lowval$ would also be strictly in the risk-neutral
part of her utility function upon winning.

For every such point, by our formulation in \eqref{eq:payid-max},
$\pricecap(\val) - \pricecap(\lowval) = \pricern(\val) -
\pricern(\val)$. As a result, since she is effectively risk-neutral in
this situation, she cannot wish to misreport $\lowval$; otherwise, the
combination of $\alloc$ and $\pricern$ would not be BIC for
risk-neutral agents.

For any $\lowval \leq \valzero$, the wealth on winning for a bidder
with value $\val$ would increase, but only into the capacitated
section of her utility function, hence gaining no utility on winning,
but losing out on a chance of winning thanks to the weak monotonicity
of $\alloc$. Hence, she would never prefer to bid $\lowval$ over
bidding $\valzero$.  Combining this argument with the above argument,
our agent with value $\val$ does not prefer to misreport any $\lowval
< \val$. \qedhere
\end{enumerate}
\end{proof}


\Xcomment{

We then argue that $\pricecap(\val) \geq \pricecap(\lowval) +
(\pricern(\val) - \pricern(\lowval))\ \ \forall \lowval < \val$. These
together give $\pricecap(\val) \geq \max(\pvc(\val), \sup \left(
\pricecap(\lowval) + (\pricern(\val) - \pricern(\lowval))\ |\ \lowval
< \val \right)$.





We assume this for the rest of the proof, that $\pricecap \geq
\pvc$. As a result, every bidder when truthfully reporting is in the
risk-neutral part of their utility function upon winning, either at
the cusp or strictly below the cusp. Given this, any such bidder is
risk-neutral when considering overbidding and paying a potentially
higher price upon winning. And as a result, if $\pricecap$ is ever
less steep than $\pricern$ at a point $\val$ - specifically, if ever
$\dpricecapright(\val) < \dpricern(\val)$ - a bidder with value $\val$
would prefer to overreport, violating BIC.

As $\pricecap$ must then always be as steep as $\pricern$, we have
$\pricecap(\val) - \pricecap(\lowval) \geq (\pricern(\val) -
\pricern(\lowval))\qquad \forall \lowval \leq \val$ and hence
$$\pricecap(\val) \geq \max(\pvc(\val), \sup \left( \pricecap(\lowval) + (\pricern(\val) - \pricern(\lowval))\ |\ \lowval < \val \right).$$





\paragraph{$\left(\pricecap(\val) \leq \max(\pvc(\val), \ldots)\right)$.} 
We consider here the case that $\pricecap(\val) > \pvc(\val)$; if
$\pricecap(\val) = \pvc(\val)$, then clearly the inequality we desire
holds. Begin by assuming for contradiction's sake that
$\pricecap(\val) > \sup \left(\pricecap(\lowval) + (\pricern(\val) -
\pricern(\lowval)) | \lowval < \val \right)$. We will show that this
must entail that the marginal cost that a bidder with value $\val$ is
paying for allocation is strictly above $\val$, and hence such a
bidder would prefer to underbid.

So, in this case, we have $\pricecap(\val) > \lim_{\lowval \to \val}
\pricecap(\lowval) + (\pricern(\val) - \pricern(\lowval))$, and as
$\alloc$ is strictly monotone, we have

\begin{align*}
\lim_{\lowval \to \val} \frac{\pricecap(\val) - \pricecap(\lowval)}{\alloc(\val) - \alloc(\lowval)} > \lim_{\lowval \to \val} \frac{\pricern(\val) - \pricern(\lowval)}{\alloc(\val) - \alloc(\lowval)}.
\end{align*} 

By the risk-neutral payment identity, we know that $\lim_{\lowval \to \val} \frac{\pricern(\val) - \pricern(\lowval)}{\alloc(\val) - \alloc(\lowval)} = \val$, that the marginal cost of an increase in allocation probability at $\val$ is $\val$. Hence, 
$\lim_{\lowval \to \val} \frac{\pricecap(\val) - \pricecap(\lowval)}{\alloc(\val) - \alloc(\lowval)} > \val$ and there exists a $\delta>0$ such that there exist arbitrarily small values $\epsilon>0$ s.t. 
\begin{align*}
\frac{\pricecap(\val) - \pricecap(\val-\epsilon)}{\alloc(\val) - \alloc(\val-\epsilon)} &\geq \val + \delta.
\end{align*}
Then, for any such $\epsilon$, we have $\alloc(\val-\epsilon)\val - \pricecap(\val-\epsilon) \geq \alloc(\val)\val - \pricecap(\val) + \delta(\alloc(\val) - \alloc(\val-\epsilon))$, and since $\alloc$ is strictly monotonic,  
\begin{align*}
\alloc(\val-\epsilon)\val - \pricecap(\val-\epsilon) &> \alloc(\val)\val - \pricecap(\val).
\end{align*} 

If the bidder with value $\val$ is in the risk-neutral part of their
utility function upon winning and reporting $\val-\epsilon$, this
states that they will strictly prefer to underreport $\val-\epsilon$
rather than truthfully reporting $\val$, violating BIC.

If $\alloc$ is continuous immediately below $\val$, then $\pricecap$
must also be continuous, and hence given that $\pricecap(\val) >
\pvc(\val)$ by assumption, there is a $\gamma >0$ s.t. for all points
$z\in [\val-\gamma, \val]$, $\pricecap(z) > (\val-\capa)\alloc(z)$ -
hence at every such point, the bidder with value $\val$ is in the
risk-neutral part of their utility function upon underbidding
$z$. Then choose an $\epsilon$ to be smaller than $\gamma$, and the
above conditions hold, giving that the bidder with value $\val$
desires to report $\val-\epsilon$, violating BIC.

If $\alloc$ is not continuous immediately below $\val$, then
$\pricecap$ can also be discontinuous immediately below $\val$. In
such a case though, if the marginal cost of the increased allocation
is higher than $\val$, then the bidder will still want to misreport
lower. In particular, for any lower point $\val-\epsilon$, we know
from above that $\pricecap(\val-\epsilon) \geq
(\val-\epsilon-\capa)\alloc(\val)$ - so on misreporting, a bidder with
value $\val$ will be really close to being risk-neutral. Thus, we can
choose an $\epsilon$ such that the utility loss from being capacitated
-
$\alloc(\val-\epsilon)(\val-\frac{\pricecap(\val-\epsilon)}{\alloc(\val-\epsilon)}
- \capa)$ is less than $\delta(\alloc(\val) - \alloc(\val-\epsilon))$,
the guaranteed difference between misreporting and truthfully
reporting. By our assumption here that $\alloc$ is not continuous
immediately below $\val$, this is feasible. Hence a bidder with value
$\val$ would still wish to misreport $\val-\epsilon$, violating BIC.
\end{proof}

}

\end{document}